\documentclass{IEEEtran}
\usepackage[T1]{fontenc}
\usepackage[latin9]{inputenc}
\usepackage{color}
\usepackage{array}
\usepackage{verbatim}
\usepackage{amsmath}
\usepackage{graphicx}

\makeatletter

\providecommand{\tabularnewline}{\\}

\newtheorem{thm}{Theorem}
\newtheorem{definitn}{Definition}
\newtheorem{cor}{Corollary}
\newtheorem{prop}{Proposition}
\newtheorem{lemma}{Lemma}

\@ifundefined{showcaptionsetup}{}{%
 \PassOptionsToPackage{caption=false}{subfig}}
\usepackage{subfig}
\makeatother

\begin{document}

\title{\vspace{0.25in}Approaching Throughput-optimality in Distributed
CSMA Scheduling Algorithms with Collisions}

\author{Libin Jiang and Jean Walrand%
\thanks{This work is supported by MURI grant BAA 07-036.18. %
}%
\thanks{Some main results of this paper appeared in \cite{S3}.%
}%
\thanks{This is the longer version of a same-titled paper to appear in \emph{IEEE/ACM
Transactions on Networking}. %
}\\
 EECS Department, University of California at Berkeley\\
 \{ljiang,wlr\}@eecs.berkeley.edu}
\maketitle
\begin{abstract}
It was shown recently that CSMA (Carrier Sense Multiple Access)-like
distributed algorithms can achieve the maximal throughput in wireless
networks (and task processing networks) under certain assumptions.
One important, but idealized assumption is that the sensing time is
negligible, so that there is no collision. In this paper, we study
more practical CSMA-based scheduling algorithms with collisions. First,
we provide a Markov chain model and give an explicit throughput formula
which takes into account the cost of collisions and overhead. The
formula has a simple form since the Markov chain is {}``almost''
time-reversible. Second, we propose transmission-length control algorithms
to approach throughput optimality in this case. Sufficient conditions
are given to ensure the convergence and stability of the proposed
algorithms. Finally, we characterize the relationship between the
CSMA parameters (such as the maximum packet lengths) and the achievable
capacity region.\end{abstract}

\section{Introduction}

Efficient resource allocation is essential to achieve high utilization
of a class of networks with resource-sharing constraints, such as
wireless networks and stochastic processing networks (SPN \cite{SPN}).
In wireless networks, certain links cannot transmit at the same time
due to the interference constraints among them. In a SPN, two tasks
cannot be processed simultaneously if they both require monopolizing
a common resource. A scheduling algorithm determines which link to
activate (or which task to process) at a given time without violating
these constraints. Designing efficient distributed scheduling algorithms
to achieve high throughput is especially a challenging task \cite{tutorial}.\global\long\def\bsig{{\bf \sigma}}
\global\long\def\br{{\bf r}}

Maximal-weight scheduling (MWS) \cite{TE92} is a classical \emph{throughput-optimal}
algorithm. That is, MWS can stabilize all queues in the network as
long as the arrival rates are within the capacity region. MWS operates
in slotted time. In each slot, a set of non-conflicting links (called
an {}``independent set'', or {}``IS'') that have the maximal weight
(i.e., summation of queue lengths) are scheduled. However, implementing
MWS in general networks is quite difficult for two reasons. (i) MWS
is inherently a centralized algorithm and is not amenable to distributed
implementation; (ii) finding a maximal-weighted IS (in each slot)
is NP-complete in general and is hard even for centralized algorithms. 

On the other hand, there has been active research on low-complexity
but suboptimal scheduling algorithms. For example, reference \cite{Maximal-Scheduling}
shows that Maximal Scheduling can only guarantee a fraction of the
network capacity. A related algorithm has been studied in \cite{bound-greedy}
in the context of IEEE 802.11 networks. Longest-Queue-First (LQF)
algorithm (see, for example, \cite{Dimakis,GMS,multi-hop-LoP,LQF-small}),
which greedily schedules queues in the descending order of the queue
lengths, tends to achieve higher throughput than Maximal Scheduling,
although it is not throughput-optimal in general \cite{GMS}. Reference
\cite{no-message-passing} proposed random-access-based algorithms
that can achieve performance comparable to that of maximal-size scheduling. 

Recently, we proposed a distributed adaptive CSMA (Carrier Sense Multiple
Access) algorithm \cite{Allerton} that is throughput-optimal for
a general interference model, under certain assumptions (further explained
later). The algorithm has a few desirable features. It is distributed
(i.e., each node only uses its own backlog information), asynchronous
(i.e., nodes do not need to synchronize their transmission) and requires
no control message. (In \cite{CISS-shah}, Rajagopalan and Shah independently
proposed a similar randomized algorithm in the context of optical
networks. Reference \cite{CSMA_marbach} showed that under a {}``node-exclusive''
interference model, CSMA can be made throughput-optimal in an asymptotic
regime.) We have also developed a joint algorithm in \cite{Allerton}
that combined the adaptive CSMA scheduling with congestion control
to approach the maximal total utility of competing data flows. %
\begin{comment}
In \cite{conv}, the optimality of these algorithms is proved formally.
\end{comment}
{}

However, the algorithms in \cite{Allerton} assume an idealized CSMA
protocol (\cite{csma-87,Kar,Thiran,BoE}), meaning that the sensing
is instantaneous, so that conflicting links do not transmit at the
same time (i.e., collisions do not occur). In many situations, however,
this is an unnatural assumption. For example, in CSMA/CA wireless
networks, due to the propagation delay and processing time, sensing
is not instantaneous. Instead, time can be viewed as divided into
discrete minislots, and collisions happen if multiple conflicting
links try to transmit in the same minislot. When a collision occurs,
all links that are involved lose their packets, and will try again
later. %
\begin{comment}
As another example, in a task processing network, if the requests
for resources are initiated in discrete slots, then there is also
an issue of collision and contention resolution.
\end{comment}
{}

In this paper, we study this important practical issue when designing
CSMA-based scheduling algorithms. We follow four main steps: (1) we
first present a Markov chain model for a simple CSMA protocol with
collisions, and give an explicit throughput formula (in section \ref{sec:formula})
that has a simple form since the Markov chain is {}``almost'' time-reversible;
(2) we show that the algorithms in \cite{Allerton} can be extended
to approach throughput optimality even with collisions (section \ref{sec:algorithms});
(3) we give sufficient conditions to ensure the convergence/stability
of the proposed algorithms (section \ref{sec:algorithms}); (4) finally,
we discuss the tradeoff between the achievable capacity region and
short-term fairness, and we characterize the relationship between
the CSMA parameters (such as the maximum packet lengths) and the achievable
capacity region. 

Although step (2) can be viewed as a generalization of \cite{Allerton},
we believe that this generalization is important and non-trivial.
The importance, as mentioned above, is because collisions are unavoidable
in CSMA/CA wireless networks, and the collision-free model used in
\cite{Allerton} does not provide enough accuracy. The generalization
is non-trivial for the following reasons. Step (2) requires the result
of step (1). In step (1), the Markov chain used to model the CSMA
protocol is no longer time-reversible as in \cite{Allerton}. We need
to re-define the state space in order to compute the service rates
it can provide. Interestingly, as a result of our design the chain
is {}``almost'' time-reversible which can be exploited. In step
(2), in view of the expression of service rates derived in step (1),
it is important to realize that adjusting the backoff times as in
\cite{Allerton} does not lead to the desirable throughput-optimal
property. Instead, one should adjust the transmission lengths. Different
from \cite{Allerton}, we further show that the {}``optimal'' CSMA
parameters (in this case the mean transmission lengths) are \emph{unique}.
This fact is needed to establish the convergence of our algorithm
in step (3).

We note that in a recent work \cite{Jian}, Ni and Srikant also proposed
a CSMA-like algorithm to achieve near-optimal throughput with collisions
taken into account. The algorithm in \cite{Jian} uses synchronized
and alternate control phase and data phase. It is designed to realize
a discrete-time CSMA (in the data phase) which has the same stationary
distribution as its continuous counterpart in \cite{Allerton}. The
control phase does not contribute to the throughput and can be viewed
as the protocol overhead. Different from \cite{Jian}, our algorithm
here is asynchronous, and has more resemblance to the RTS/CTS mode
in IEEE 802.11. Although the algorithm in \cite{Jian} is quite elegant
and could potentially be applied to other time-slotted systems, we
believe that it is an interesting problem to understand how to use
the asynchronous algorithm to achieve throughput-optimality. Also,
due to its similarity to the RTS/CTS mode of IEEE 802.11, the throughput
analysis in this paper could also deepen the understanding of 802.11
in general topologies.

\section{\label{sec:formula}CSMA/CA-based scheduling with collisions}

\subsection{Model}

In this section we present a model for CSMA/CA-based scheduling with
collisions. Note that the goal of the paper is \emph{not} to propose
a comprehensive model that captures all details for IEEE 802.11 networks
and predict the performance of such networks (The literature in that
area has been very rich. See, for example, \cite{new_insight,Kar}
and the references therein.) Instead, at a more abstract level, we
are interested in a distributed scheduling algorithm that is inspired
by CSMA/CA, and designing adaptive algorithms to approach throughput-optimality.
%
\begin{comment}
Such algorithms can be used not only in wireless networks, but also
in a general task processing problem.

\subsubsection{A model in the context of CSMA/CA wireless networks}
\end{comment}
{}

Consider a (single-channel) wireless network. Define a {}``link''
as an (ordered) transmitter-receiver pair. Assume that there are $K$
links, and denote the set of links by ${\cal N}$ (then, $K=|{\cal N}|$).
Without loss of generality, assume that each link has a capacity of
1. We say that two links \emph{conflict} if they cannot transmit (or,
{}``be active'') at the same time due to interference. (The conflict
relationship is assumed to be symmetric.) Accordingly, define $G$
as the conflict graph. Each vertex in $G$ represents a link, and
there is an edge between two vertexes if the corresponding links conflict.
Note that this simple conflict model may not reflect all possible
interference structures that could occur in wireless networks. However,
it does provide a useful abstraction and has been used widely in literature
(see, for example, \cite{Kar} and \cite{tutorial}). %
\begin{comment}
Also, it is general enough to be used in other resource sharing problems.
An example is provided in the next section.
\end{comment}
{} %
\begin{comment}
Further assume that there is no loss of packets due to channel fading.
\end{comment}
{}

Fig. \ref{fig:WLAN} (a) shows a wireless LAN with 6 links. The network's
conflict graph is a full graph (Fig. \ref{fig:WLAN} (b)). (Circles
represent nodes and rectangles represent links.) Fig. \ref{fig:Ad-hoc}
(a) shows an ad-hoc network with 3 links. Assume that link 1, 2 conflict,
and link 2, 3 conflict. Then the network's conflict graph is Fig.
\ref{fig:Ad-hoc} (b). %
\begin{figure}
\begin{centering}
\subfloat[Network topology]{\includegraphics[clip,width=3cm]{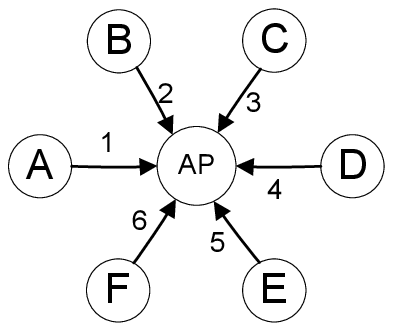} 

} \subfloat[Conflict graph]{\includegraphics[clip,width=3.5cm]{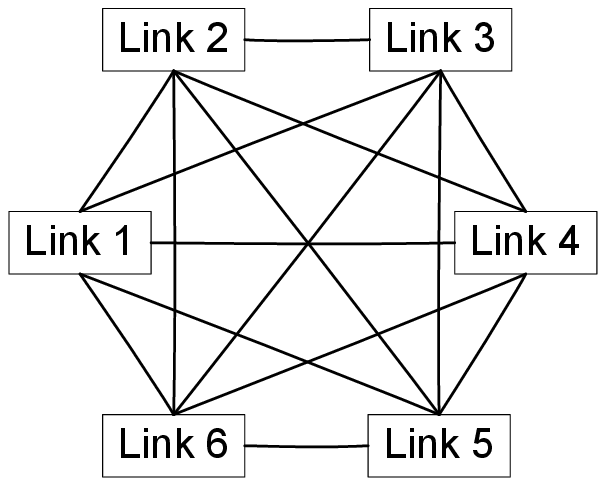}

}
\par\end{centering}

\caption{\label{fig:WLAN}Infrastructure network}

\end{figure}
\begin{figure}
\begin{centering}
\subfloat[Network topology]{\includegraphics[clip,width=3cm]{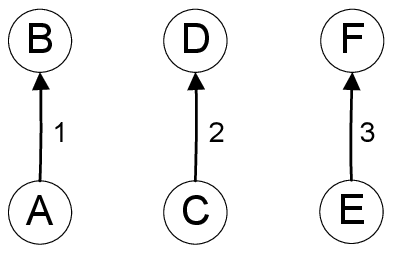} 

} \subfloat[Conflict graph]{\includegraphics[clip,width=4cm]{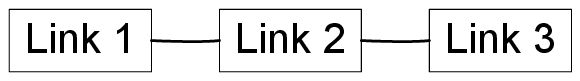}

}
\par\end{centering}

\caption{\label{fig:Ad-hoc}Ad-hoc network}

\end{figure}
\medskip{}

{\bf Basic Protocol}\medskip{}

We now describe the basic CSMA/CA protocol with fixed transmission
probabilities (which suffices for our later development.) Let $\tilde{\sigma}$
be the duration of each idle slot (or {}``minislot''). ($\tilde{\sigma}$
should be at least the time needed by any wireless station to detect
the transmission of any other station. Specifically, it accounts for
the propagation delay, the time needed to switch from the receiving
to the transmitting state, and the time to signal to the MAC layer
the state of the channel \cite{Bianchi}. The value of $\tilde{\sigma}$
varies for different physical layers. In IEEE 802.11a, for example,
$\tilde{\sigma}=9\mu s$.) In the following we will simply use {}``slot''
to refer to the minislot. 

Assume that all links are saturated (i.e., always have packets to
transmit). In each slot, if (the transmitter of) link $i$ is not
already transmitting and if the medium is idle, the transmitter of
link $i$ starts transmission with probability $p_{i}$. If at a certain
slot, link $i$ did not choose to transmit but a conflicting link
starts transmitting, then link $i$ keeps silent until that transmission
ends. If conflicting links start transmitting at the same slot, then
a collision happens. {[}We assume that the network has no hidden node
(HN). The case with HNs will be discussed in Section \ref{sub:HN}.
For possible ways to address the HN problem, please refer to \cite{Jiang-Liew}
and its references.] 

Each link transmits a short probe packet with length $\gamma$ (similar
to the RTS packet in 802.11) before the data is transmitted. (All
{}``lengths'' here are measured in number of slots and are assumed
to be integers.) This increases the overhead of successful transmissions,
but can avoid collisions of long data packets. When a collision happens,
only the probe packets collide, so each collision lasts a length of
$\gamma$. Assume that a successful transmission of link $i$ lasts
$\tau_{i}$, which includes a constant overhead $\tau'$ (composed
of RTS, CTS, ACK, etc) and the data payload $\tau_{i}^{p}$ which
is a random variable. Clearly $\tau_{i}\ge\tau'$. Let the p.m.f.
(probability mass function) of $\tau_{i}$ be \begin{equation}
Pr\{\tau_{i}=b_{i}\}=P_{i}(b_{i}),\forall b_{i}\in{\cal Z}_{++}\label{eq:pmf}\end{equation}
and assume that the p.m.f. has a \emph{finite support}, i.e., $P_{i}(b_{i})=0,\forall b_{i}>b_{max}>0$.\textcolor{blue}{}%
\footnote{The finite support assumption is not a restrictive one, since in practical
wireless networks there is usually an upper bound on the packet size.
However, the CSMA/CA Markov chain defined later is still ergodic even
without the assumption.%
}\textcolor{blue}{ }Then the mean of $\tau_{i}$ is \begin{equation}
T_{i}:=E(\tau_{i})=\sum_{b\in{\cal Z}_{++}}b\cdot P_{i}(b).\label{eq:T_i}\end{equation}

Fig. \ref{fig:basic-model} illustrates the timeline of the 3-link
network in Fig. \ref{fig:Ad-hoc} (b), where link 1 and 2 conflict,
and link 2 and 3 conflict.

We note a subtle point in our modeling. In IEEE 802.11, a link can
attempt to start a transmission only after it has sensed the medium
as idle for a constant time (which is called DIFS, or {}``DCF Inter
Frame Space''). To take this into account, DIFS is included in the
packet transmission length $\tau_{i}$ and the collision length $\gamma$.
In particular, for a successful transmission of link $i$, DIFS is
included in the constant overhead $\tau'$. Although DIFS, as part
of $\tau'$, is actually after the payload, in Fig. \ref{fig:basic-model}
we plot $\tau'$ before the payload. This is for convenience and does
not affect our results. So, under this model, a link can attempt to
start a transmission \emph{immediately} after the transmissions of
its conflicting links end.

The above model is almost time-reversible such that a simple throughput
formula can be derived. A process is {}``time-reversible'' if the
process and its time-reversed process are statistically indistinguishable
\cite{Kelly-book}. Our model, in Fig. \ref{fig:basic-model}, reversed
in time, follows the same protocol as described above, except for
the order of the overhead and the payload, which are reversed. A key
reason for this property is that the collisions start and finish at
the same time. (This point will be made more precise in Appendix A.)

\begin{figure}
\begin{centering}
\includegraphics[clip,width=8cm]{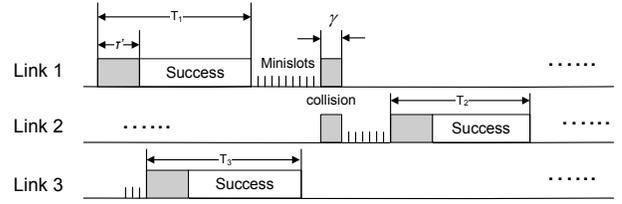}
\par\end{centering}

\caption{\label{fig:basic-model}Timeline in the basic model (In this figure,
$\tau_{i}=T_{i},i=1,2,3$ are constants.)}

\end{figure}

{}

\subsection{Notation}

Let the {}``on-off state'' be $x\in\{0,1\}^{K}$ where $x_{k}$,
the $k$-th element of $x$, is such that $x_{k}=1$ if link $k$
is active (transmitting) in state $x$, and $x_{k}=0$ otherwise.
Thus, $x$ is a vector indicating which links are active in a given
slot. Let $G(x)$ be the subgraph of $G$ after removing all vertices
(each representing a link) with state 0 (i.e., any link $j$ with
$x_{j}=0$) and their associated edges. In general, $G(x)$ is composed
of a number of connected components (simply called {}``components'')
$C_{m}(x),m=1,2,\dots,M(x)$ (where each component is a set of links,
and $M(x)$ is the total number of components in $G(x)$). If a component
$C_{m}(x)$ has only one active link (i.e., $|C_{m}(x)|=1$), then
this link is having a successful transmission; if $|C_{m}(x)|>1$,
then all the links in the component are experiencing a collision.
Let the set of {}``successful'' links in state $x$ be $S(x):=\{k|k\in C_{m}(x)\text{ with }|C_{m}(x)|=1\}$,
and the set of links that are experiencing collisions be $Z(x)$.
Also, define the {}``collision number'' $h(x)$ as the number of
components in $G(x)$ with size larger than 1. Fig. \ref{fig:Example_c}
shows an example. Note that the transmissions in a collision component
$C_{m}(x)$ are {}``synchronized'', i.e., the links in $C_{m}(x)$
must have started transmitting in the same slot, and will end transmitting
in the same slot after $\gamma$ slots (the length of the probe packets).%
\begin{figure}
\begin{centering}
\includegraphics[clip,width=4.5cm]{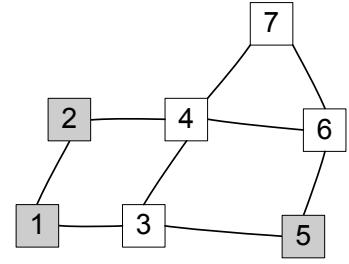}
\par\end{centering}

\caption{\label{fig:Example_c}An example conflict graph (each square represents
a link). In this on-off state $x$, links 1, 2, 5 are active. So $S(x)=\{5\}$,
$Z(x)=\{1,2\}$, $h(x)=1$.}

\end{figure}

\subsection{Computation of the service rates}

In order to compute the service rates of all the links under the above
CSMA protocol when all the links are saturated, we first define the
underlying discrete-time Markov chain which we call the \emph{CSMA/CA
Markov chain}.

The Markov chain evolves slot by slot.%
\footnote{For the ease of analysis, we make the modeling assumption that the
links are synchronized at the slot level.%
} The state of the Markov chain in a slot is \begin{equation}
w:=\{x,((b_{k},a_{k}),\forall k:x_{k}=1)\}\label{eq:w-definition}\end{equation}
where $b_{k}$ is the total length of the current packet link $k$
is transmitting, $a_{k}$ is the remaining time (including the current
slot) before the transmission of link $k$ ends. 

For example, in Fig. \ref{fig:Example-Markov}, the state $w$ and
$w'$ are%
\begin{figure}
\begin{centering}
\includegraphics[width=6cm]{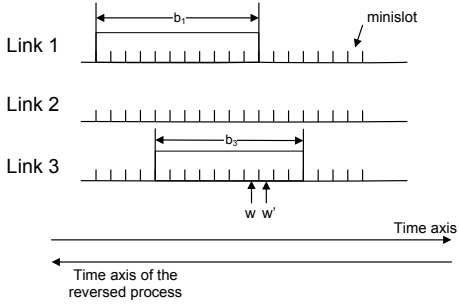}
\par\end{centering}

\caption{\label{fig:Example-Markov}Example of the CSMA/CA Markov chain}

\end{figure}
\begin{equation}
w=\{x=(1,0,1)^{T},(b_{1}=11,a_{1}=1),(b_{3}=10,a_{3}=4)\}\label{eq:w-example}\end{equation}
and\begin{equation}
w'=\{x=(0,0,1)^{T},(b_{3}'=10,a_{3}'=3)\}.\label{eq:wp-example}\end{equation}

Note that in any state $w$ as defined in (\ref{eq:w-definition}),
we have

(I) $1\le a_{k}\le b_{k},\forall k:x_{k}=1$. 

(II) $P_{k}(b_{k})>0,\forall k\in S(x)$.

(III) If $k\in Z(x)$, then $b_{k}=\gamma$ and $a_{k}\in\{1,2,\dots,\gamma\}$.
An important observation here is that the transmissions in a collision
component $C_{m}(x)$ are {}``synchronized'', i.e., the links in
$C_{m}(x)$ must have started transmitting at the same time, and will
end transmitting at the same time, so all links in the component $C_{m}(x)$
have the same remaining time. (To see this, first note that in the
case of a collision only the probe packets get transmitted, and their
transmission times $\gamma$ are identical for all links. Second,
any two links $i$ and $j$ in this component with an edge between
them must have started transmitting at the same time. Otherwise, if
$i$ starts earlier, $j$ would not transmit since it already hears
$i$'s transmission; and vice versa. By induction, all links in the
component must have started transmitting at the same time.) So, we
can write $a_{k}=a^{(m)}$ for any $k\in C_{m}(x)$ where $|C_{m}(x)|>1$,
and $a^{(m)}$ denotes the remaining time of the component $C_{m}(x)$. 

We say that a state $w$ is \emph{valid} iff it satisfies (I)\textasciitilde{}(III)
above.

Since the transmission lengths are always bounded by $b_{max}$ by
assumption, we have $b_{k}\le b_{max}$, and therefore the Markov
chain has a finite number of states, and is ergodic. As detailed in
Appendix \ref{sub:Proof-product-form}, a nice property of this Markov
chain is that it is {}``almost'' time-reversible. As a result, its
stationary distribution has a simple product-form (Appendix \ref{sub:Proof-product-form}),
from which the probability of any on-off state $x$ can be computed:
\begin{thm}
\label{thm:state-distribution}Under the stationary distribution,
the probability of $x\in\{0,1\}^{K}$ in a given slot is \begin{eqnarray}
p(x) & = & \frac{1}{E}(\gamma^{h(x)}\prod_{k\in S(x)}T_{k})\prod_{i:x_{i}=0}(1-p_{i})\prod_{j:x_{j}=1}p_{j}\nonumber \\
 & = & \frac{1}{E}(\gamma^{h(x)}\prod_{k\in S(x)}T_{k})\prod_{i=1}^{K}p_{i}^{x_{i}}q_{i}^{1-x_{i}}\label{eq:x}\end{eqnarray}
where $q_{i}:=1-p_{i}$, $T_{i}$ is the mean transmission length
of link $i$ (as defined in (\ref{eq:T_i})), and $E$ is a normalizing
term such that $\sum_{x\in\{0,1\}^{K}}p(x)=1$.%
\footnote{In this paper, several kinds of {}``states'' are defined. With a
little abuse of notation, we always use $p(\cdot)$ to denote the
probability of the {}``state'' under the stationary distribution
of the CSMA/CA Markov chain. This does not cause confusion since the
meaning of $p(\cdot)$ is clear from its argument.%
}

The proof is given in Appendix A.%
\footnote{In \cite{asymptotic}, a similar model for CSMA/CA network is formulated
with analogy to a loss network \cite{loss_network}. However, since
\cite{asymptotic} studied the case when the links are unsaturated,
the explicit expression of the stationary distribution was difficult
to obtain.%
}

Remark: Note that in $x$, some links can be in a collision state,
just as in IEEE 802.11. This is reflected in the $\gamma^{h(x)}$
term in (\ref{eq:x}). Expression (\ref{eq:x}) differs from the idealized-CSMA
case in \cite{Allerton} and the stationary distribution in the data
phase in the protocol proposed in \cite{Jian}, due to the difference
of the three protocols.
\end{thm}
Now we re-parametrize $T_{k}$ by a variable $r_{k}$. Let $T_{k}:=\tau'+T_{0}\cdot\exp(r_{k})$,
where $\tau'$, as we defined, is the overhead of a successful transmission
(including RTS, CTS, ACK packets, DIFS, etc.), and $T_{k}^{p}:=T_{0}\cdot\exp(r_{k})$
is the mean length of the payload. $T_{0}>0$ is a constant {}``reference
payload length''. Let ${\bf r}$ be the vector of $r_{k}$'s. By
Theorem \ref{thm:state-distribution}, the stationary probability
of $x$ in a slot (with a given ${\bf r}$) is \begin{equation}
p(x;{\bf r})=\frac{1}{E({\bf r})}g(x)\cdot\prod_{k\in S(x)}(\tau'+T_{0}\cdot\exp(r_{k}))\label{eq:px_given_r}\end{equation}
where $g(x)=\gamma^{h(x)}\prod_{i=1}^{K}p_{i}^{x_{i}}q_{i}^{1-x_{i}}$
is not related to ${\bf r}$, and the normalizing term is \begin{equation}
E({\bf r})=\sum_{x'\in\{0,1\}^{K}}[g(x')\cdot\prod_{k\in S(x')}(\tau'+T_{0}\cdot\exp(r_{k}))].\label{eq:E_r}\end{equation}

Then, the stationary probability that link $k$ is transmitting a
payload in a given slot is\begin{equation}
s_{k}({\bf r})=\frac{T_{0}\cdot\exp(r_{k})}{\tau'+T_{0}\cdot\exp(r_{k})}\sum_{x:k\in S(x)}p(x;{\bf r}).\label{eq:throughput}\end{equation}
Recall that the capacity of each link is 1. Also, it's easy to show
that the CSMA/CA Markov chain is ergodic. As a result, if ${\bf r}$
is fixed, the long-term average throughput of link $k$ converges
to the stationary probability $s_{k}({\bf r})$. So we say that $s_{k}({\bf r})\in[0,1]$
is the \emph{service rate} of link $k$.

\section{\label{sec:algorithms}A Distributed algorithm to approach throughput-optimality}

\subsection{The scheduling problem}

Assume that the conflict graph $G$ has $N$ different independent
sets ({}``IS'', not confined to {}``maximal independent sets''),
where each IS is a set of links that can transmit simultaneously without
conflict. Denote an IS by $\bsig\in\{0,1\}^{K}$, a 0-1 vector that
indicates which links are transmitting in this IS. The $k$'th element
of $\bsig$, $\sigma_{k}=1$ if link $k$ is transmitting in this
IS, and $\sigma_{k}=0$ otherwise. Let ${\cal X}$ be the set of ISs. 

We now describe the \emph{scheduling problem} which is the focus of
the paper. Traffic arrives at link $k$ with an arrival rate $\lambda_{k}\in(0,1)$.
For simplicity, assume the following i.i.d. Bernoulli arrivals (although
this can be easily generalized): at the beginning of slot $M\cdot i$,
$(i=0,1,2,\dots)$, a packet with a length of $M$ slots arrives at
link $k$ with probability $\lambda_{k}$. (That is, the packet would
take $M$ slots to transmit.) Clearly, link $k$ needs to be active
with a probability at least $\lambda_{k}$ to serve the arrivals.
Denote the vector of arrival rates by ${\bf \lambda}\in{\cal R}_{+}^{K}$. 

Since we focus on the scheduling problem, all the packets traverse
only one link (i.e., \emph{single-hop}) before they leave the network.
However, the results here can be extended to multi-hop networks and
be combined with congestion control as in \cite{Allerton}.
\begin{definitn}
\label{def:feasible}We say that ${\bf \lambda}$ is\emph{ feasible}
iff it can be written as ${\bf \lambda}=\sum_{\bsig\in{\cal X}}[\bar{p}_{\bsig}\cdot\bsig]$
where $\bar{p}_{\bsig}\ge0$ and $\sum_{\bsig\in{\cal X}}\bar{p}_{\bsig}=1$.
That is, there is a schedule of the independent sets (including the
non-maximal ones) that can serve the arrivals. Denote the set of feasible
${\bf \lambda}$ by $\bar{{\cal C}}$. We say that ${\bf \lambda}$
is \emph{strictly feasible} iff $\lambda\in{\cal C}$, where ${\cal C}$
is the interior of $\bar{{\cal C}}$. %
\footnote{That is, ${\cal C}:=\{{\bf \lambda}'\in\bar{{\cal C}}|{\cal B}({\bf \lambda}',d)\subseteq\bar{{\cal C}}\text{ for some }d>0\}$,
where ${\cal B}({\bf \lambda}',d)=\{\tilde{{\bf \lambda}}|\:||\tilde{{\bf \lambda}}-{\bf \lambda}'||_{2}\le d\}$
is a ball centered at ${\bf \lambda}'$ with radius $d$.%
}
\end{definitn}
A scheduling algorithm is said to be {}``throughput-optimal'' if
it can {}``support'' any $\lambda\in{\cal C}$. In this paper, this
means that for any ${\bf \lambda}\in{\cal C}$, the scheduling algorithm
can provide to link $k$ a service rate at least $\lambda_{k}$ for
all $k$.

\subsection{CSMA scheduling with collisions}

The following theorem states that any service rates equal to ${\bf \lambda}\in{\cal C}$
can be achieved by properly choosing the mean payload lengths $T_{k}^{p}:=T_{0}\exp(r_{k}),\forall k$.
\begin{thm}
\label{thm:thu-optimal}Assume that $\gamma,\tau'>0$, and transmission
probabilities $p_{k}\in(0,1),\forall k$ are fixed. Given any ${\bf \lambda}\in{\cal C}$,
there exists a unique ${\bf r}^{*}\in{\cal R}^{K}$ such that the
service rate of link $k$ is equal to the arrival rate for all $k$:
\begin{equation}
s_{k}({\bf r}^{*})=\lambda_{k},\forall k.\label{eq:thu-optimal}\end{equation}

Moreover, ${\bf r}^{*}$ is the solution of the convex optimization
problem\begin{equation}
\max_{{\bf r}}L({\bf r};{\bf \lambda})\label{eq:loglikelihood-primal_c}\end{equation}
where \begin{equation}
L({\bf r};{\bf \lambda})=\sum_{k}(\lambda_{k}r_{k})-\log(E({\bf r})),\label{eq:Lr}\end{equation}
with $E({\bf r})$ defined in (\ref{eq:E_r}). This is because $\partial L({\bf r};{\bf \lambda})/\partial r_{k}=\lambda_{k}-s_{k}({\bf r}),\forall k$.

The proof is in Appendix B.
\end{thm}
Theorem \ref{thm:thu-optimal} motivates us to design a gradient algorithm
to solve problem (\ref{eq:loglikelihood-primal_c}). However, due
to the randomness of the system, $\lambda_{k}$ and $s_{k}({\bf r})$
cannot be obtained directly and need to be estimated. We design the
following distributed algorithm, where each link $k$ dynamically
adjusts its mean payload length $T_{k}^{p}$ based on local information.
\smallskip{}

\textbf{Algorithm 1: Transmission length control algorithm}

\medskip{}

The vector ${\bf r}$ is updated every $M$ slots. Specifically, it
is updated at the beginning of slot $M\cdot i$, $i=1,2,\dots$. Denote
by $t_{i}$ the time when slot $M\cdot i$ begins. Also define $t_{0}=0$.
Let {}``period $i$'' be the time between $t_{i-1}$ and $t_{i}$,
and ${\bf r}(i)$ be the value of ${\bf r}$ at the end of period
$i$, i.e., at time $t_{i}$. Initially, link $k$ sets $r_{k}(0)\in[r_{min},r_{max}]$
where $r_{min},r_{max}$ are two parameters (to be further discussed).
Then at time $t_{i}$, $i=1,2,\dots$, each link $k$ updates \begin{equation}
r_{k}(i)=r_{k}(i-1)+\alpha(i)[\lambda_{k}'(i)-s_{k}'(i)+\tilde{h}(r_{k}(i-1))]\label{eq:Alg1}\end{equation}
where $\tilde{h}(\cdot)$ is a {}``penalty function'' to be defined
later, $\alpha(i)>0$ is the step size in period $i$, $\lambda_{k}'(i),s_{k}'(i)$
are the empirical average arrival rate and service rate in period
$i$ (i.e., the actual amount of arrived traffic and served traffic
in period $i$ divided by $M$). 

\textbf{The use of dummy bits:}\emph{ An important point} here is
that we let link $k$ add dummy bits to the payload when its queue
has less bits than what is specified by the algorithm (e.g., in (\ref{eq:randomize-Tp})
below). If the queue is empty, then dummy packets are transmitted
with the specified size. So, each link is \emph{saturated}. This ensures
that the CSMA/CA Markov chain has the desired stationary distribution
in (\ref{eq:x}). The transmitted dummy bits are also included in
the computation of $s_{k}'(i)$. (Although the use of dummy bits consumes
bandwidth, it simplifies our analysis, and does not prevent us from
achieving the primary goal, i.e., approaching throughput-optimality.
In Section \ref{sub:dummy_bits}, we also simulate the case without
dummy bits.) 

Note that $\lambda_{k}'(i),s_{k}'(i)$ are random variables which
are generally not equal to $\lambda_{k}$ and $s_{k}({\bf r}(i-1))$.
Assume that the maximal instantaneous arrival rate is $\bar{\lambda}$,
so $\lambda_{k}'(i)\le\bar{\lambda},\forall k,i$. 

Also, in (\ref{eq:Alg1}), the penalty function $\tilde{h}(\cdot)$
is defined as \begin{equation}
\tilde{h}(y)=\begin{cases}
r_{min}-y & \text{if }y<r_{min}\\
0 & \text{if }y\in[r_{min},r_{max}]\\
r_{max}-y & \text{if }y>r_{max}.\end{cases}\label{eq:h}\end{equation}
As shown in the Appendix, this function keeps ${\bf r}(i)$ in a bounded
region. (This is a {}``softer'' approach than directly projecting
$r_{k}(i)$ to the set $[r_{min},r_{max}]$. The purpose is only to
simplify the proof of Theorem \ref{thm:scheduling} later.)

In period $i+1$, given ${\bf r}(i)$, we need to choose $\tau_{k}^{p}(i)$,
the payload lengths of each link $k$, so that $E(\tau_{k}^{p}(i))=T_{k}^{p}(i)=T_{0}\exp(r_{k}(i))$.
If $T_{k}^{p}(i)$ is an integer, then we let $\tau_{k}^{p}(i)=T_{k}^{p}(i)$;
otherwise, we randomize $\tau_{k}^{p}(i)$ as follows: \begin{equation}
\tau_{k}^{p}(i)=\begin{cases}
\left\lceil T_{k}^{p}(i)\right\rceil  & \text{with probability }T_{k}^{p}(i)-\left\lfloor T_{k}^{p}(i)\right\rfloor \\
\left\lfloor T_{k}^{p}(i)\right\rfloor  & \text{with probability }\left\lceil T_{k}^{p}(i)\right\rceil -T_{k}^{p}(i).\end{cases}\label{eq:randomize-Tp}\end{equation}
Here, for simplicity, we have assumed that the arrived packets can
be fragmented and reassembled to obtain the desired lengths $\left\lceil T_{k}^{p}(i)\right\rceil $
or $\left\lfloor T_{k}^{p}(i)\right\rfloor $. However, \emph{one
can avoid the fragmentation} by randomizing the number of transmitted
packets (each with a length of $M$ slots) in a similar way. When
there are not enough bits in the queue, {}``dummy bits'' are generated
(as mentioned before) to satisfy $E(\tau_{k}^{p}(i))=T_{0}\exp(r_{k}(i))$
and make the links always saturated. \medskip{}

Intuitively speaking, Algorithm 1 says that when $r_{k}\in[r_{min},r_{max}]$,
if the empirical arrival rate of link $k$ is larger than the service
rate, then link $k$ should transmit more aggressively by using a
larger mean transmission length, and vice versa. 

Algorithm 1 is parametrized by $r_{min},r_{max}$ which are fixed
during the execution of the algorithm. Note that the choice of $r_{max}$
affects the maximal possible payload length. Also, as discussed below,
the choices of $r_{max}$ and $r_{min}$ also determine the {}``capacity
region'' of Algorithm 1.

We define the region of arrival rates\begin{equation}
{\cal C}(r_{min},r_{max}):=\{{\bf \lambda}\in{\cal C}|{\bf r}^{*}({\bf \lambda})\in(r_{min},r_{max})^{K}\}\label{eq:C_min_max}\end{equation}
where ${\bf r}^{*}({\bf \lambda})$ denotes the unique solution of
$\max_{{\bf r}}L({\bf r};{\bf \lambda})$ (such that $s_{k}({\bf r}^{*})=\lambda_{k},\forall k$,
by Theorem \ref{thm:thu-optimal}). Later we show that the algorithm
can {}``support'' any ${\bf \lambda}\in{\cal C}(r_{min},r_{max})$
in some sense under certain conditions on the step sizes. We will
also give a characterization of the region ${\cal C}(r_{min},r_{max})$
later in section \ref{sec:capacity-region}.

Clearly, ${\cal C}(r_{min},r_{max})\rightarrow{\cal C}$ as $r_{min}\rightarrow-\infty$
and $r_{max}\rightarrow\infty$, where ${\cal C}$ is the set of all
strictly feasible ${\bf \lambda}$ (by Theorem \ref{thm:thu-optimal}).
Therefore, although given $(r_{min},r_{max})$ the region ${\cal C}(r_{min},r_{max})$
is smaller than ${\cal C}$, one can choose $(r_{min},r_{max})$ to
arbitrarily approach the maximal capacity region ${\cal C}$. Also,
there is a tradeoff between the capacity region and the maximal packet
length.

\medskip{}

\begin{thm}
\label{thm:scheduling}Assume that the vector of arrival rates ${\bf \lambda}\in{\cal C}(r_{min},r_{max})$.
With Algorithm 1, 

(i) If $\alpha(i)>0$ is non-increasing and satisfies $\sum_{i}\alpha(i)=\infty$,
$\sum_{i}\alpha(i)^{2}<\infty$ and $\alpha(1)\le1$ (for example,
$\alpha(i)=1/i$), then ${\bf r}(i)\rightarrow{\bf r}^{*}$ as $i\rightarrow\infty$
with probability 1, where ${\bf r}^{*}$ satisfies $s_{k}({\bf r}^{*})=\lambda_{k},\forall k$.

(ii) The case with constant step size (i.e., $\alpha(i)=\alpha,\forall i$):
For any $\delta>0$, there exists a small enough $\alpha>0$ such
that $\liminf_{J\rightarrow\infty}\sum_{i=1}^{J}s_{k}'(i)/J]\ge\lambda_{k}-\delta,\forall k$
with probability 1. In other words, one can achieve average service
rates arbitrarily close to the arrival rates by choosing $\alpha$
small enough.
\end{thm}
\emph{Remark}: In \cite{Jian} which proposed an alternative algorithm
to deal with collisions, the authors made an idealized time-scale-separation
assumption that the CSMA/CA Markov chain reaches its stationary distribution
for any given CSMA parameters. We believe that the results in Theorem
\ref{thm:scheduling} can be extended to their algorithm.

The complete proof of Theorem \ref{thm:scheduling} is Appendix C,
but the result can be intuitively understood as follows. If the step
size is small (in (i), $\alpha(i)$ becomes small when $i$ is large),
$r_{k}$ is {}``quasi-static'' such that roughly, the service rate
is averaged (over multiple periods) to $s_{k}({\bf r})$, and the
arrival rate is averaged to $\lambda_{k}$. Thus the algorithm solves
the optimization problem (\ref{eq:loglikelihood-primal_c}) by a stochastic
approximation \cite{Borkar} argument, such that ${\bf r}(i)$ converges
to ${\bf r}^{*}$ in part (i), and $r(i)$ is near ${\bf r}^{*}$
with high probability in part (ii). %
\begin{comment}
(In part (i), it can be further shown that the system is rate stable
\cite{longer_version}.) 
\end{comment}
{}

\medskip{}

\begin{cor}
\label{cor:stable}Consider a variant of Algorithm 1 below where the
update equation of each link $k$ is \begin{equation}
r_{k}(i)=r_{k}(i-1)+\alpha(i)[\lambda_{k}'(i)+\Delta-s_{k}'(i)+\tilde{h}(r_{k}(i-1))]\label{eq:Alg1b}\end{equation}
with a small constant $\Delta>0$. That is, the algorithm {}``pretends''
to serve the arrival rate ${\bf \lambda}+\Delta\cdot{\bf 1}$ which
is slightly larger than the actual ${\bf \lambda}$. Assume that \begin{eqnarray*}
{\bf \lambda} & \in & {\cal C}'(r_{min},r_{max},\Delta)\\
 & := & \{{\bf \lambda}|{\bf \lambda}+\Delta\cdot{\bf 1}\in{\cal C}(r_{min},r_{max})\}.\end{eqnarray*}

For algorithm (\ref{eq:Alg1b}), one has the following results:

(i) if $\alpha(i)>0$ is non-increasing and satisfies $\sum_{i}\alpha(i)=\infty$,
$\sum_{i}\alpha(i)^{2}<\infty$ and $\alpha(1)\le1$ (for example,
$\alpha(i)=1/i$), then ${\bf r}(i)\rightarrow{\bf r}^{*}$ as $i\rightarrow\infty$
with probability 1, where ${\bf r}^{*}$ satisfies $s_{k}({\bf r}^{*})=\lambda_{k}+\Delta>\lambda_{k},\forall k$;

(ii) if $\alpha(i)=\alpha$ (i.e., constant step size) where $\alpha$
is small enough, then all queues are positive recurrent (and therefore
stable). 

Algorithm (\ref{eq:Alg1b}) is parametrized by $r_{min},r_{max}$
and $\Delta$. Clearly, as $r_{min}\rightarrow-\infty$, $r_{max}\rightarrow\infty$
and $\Delta\rightarrow0$, ${\cal C}'(r_{min},r_{max},\Delta)\rightarrow{\cal C}$,
the maximal capacity region.

The proof is similar to that of Theorem \ref{thm:scheduling} and
is given in \cite{longer_version}. A sketch is as follows: Part (i)
is similar to (i) in Theorem \ref{thm:scheduling}. The extra fact
that $s_{k}({\bf r}^{*})>\lambda_{k},\forall k$ reduces the queue
size compared to Algorithm 1 (since when the queue size is large enough,
it tends to decrease). Part (ii) holds because if we choose $\delta=\Delta/2$,
then by Theorem \ref{thm:scheduling}, $\liminf_{J\rightarrow\infty}\sum_{i=1}^{J}s_{k}'(i)/J]\ge\lambda_{k}+\Delta-\delta>\lambda_{k},\forall k$
almost surely if $\alpha$ is small enough. Then the result follows
by showing that the queue sizes have negative drift. 
\end{cor}

\section{\label{sec:capacity-region}Relationship between the CSMA parameters
and the capacity region}

In the previous section, we mentioned that the region ${\cal C}(r_{min},r_{max})$
(and ${\cal C}'(r_{min},r_{max},\Delta)$) becomes larger as we decrease
$r_{min}$ and/or increase $r_{max}$. Therefore, fixing $r_{min}$,
a larger $r_{max}$ leads to a larger capacity region, but allows
for larger transmission lengths. In practice, however, the transmission
lengths should not be too long, since longer transmission lengths\textbf{
}lead to larger access delay (where the access delay refers to the
time between the beginnings of two consecutive successful transmissions
of a link) and larger variations of the delay, and consequently, poorer
\emph{short-term fairness}. It is especially the case when a link
has a number of conflicting links which do not interfere with each
other. Then the link has to wait for all the conflicting links to
become inactive before attempting its transmission. This issue has
been studied in \cite{Thiran}\cite{conv_MS}\cite{BoE} in the contexts
of 1-D and 2-D lattice topologies, and star topologies, where it is
shown that the short-term fairness worsens when the \emph{access intensities}
(i.e., the ratios between the average transmission times and mean
backoff times) increase%
\footnote{Reference \cite{Thiran} also showed that when the access intensities
are high, there exists long-term unfairness in the 2-D lattice topology
under different boundary conditions.%
}. Although references \cite{Thiran}\cite{conv_MS}\cite{BoE} focus
on the collision-free idealized-CSMA, we observe the same phenomenon
in the simulations of our model (see Appendix E for some simulation
results in the 1-D and 2-D lattice topologies). 

Therefore, there is a \emph{tradeoff} between the long-term efficiency
(i.e., the capacity region) and short-term fairness. To quantify the
tradeoff we need to understand two relationships. The first is the
relationship between the maximal required transmission lengths and
the capacity region. And the second is between the maximal transmission
lengths and the short-term fairness. 

We first discuss the second relationship. For simplicity, assume the
arrival rate vector is ${\bf \lambda}$, and that Algorithm 1 has
converged to the suitable mean payload lengths $T_{k}^{p}:=T_{0}\exp(r_{k}^{*}({\bf \lambda})),\forall k$
(Recall that ${\bf r}^{*}({\bf \lambda})$ is the vector such that
$s_{k}({\bf r}^{*}({\bf \lambda}))=\lambda_{k},\forall k$). Assume
that we fix the mean payload lengths at $T_{k}^{p}$'s, and denote
the (random) access delay of link $k$ by $D_{k}$. We use two quantities
to measure the short-term fairness of link $k$: the mean and standard
deviation of $D_{k}$ (i.e., $E(D_{k})$ and $\sqrt{var(D_{k})}$).
Similar to \cite{conv_MS}, one has\[
E(D_{k})=\frac{T_{k}^{p}}{s_{k}({\bf r}^{*}({\bf \lambda}))}=\frac{T_{k}^{p}}{\lambda_{k}},\forall k.\]
Therefore if we can find an upper bound of $T_{k}^{p}=T_{0}\exp(r_{k}^{*}({\bf \lambda}))$
(to be further discussed in this section), then an upper bound of
$E(D_{k})$ can be obtained. (In fact, in Algorithm 1, the long-term
average of $D_{k}$ is also $T_{k}^{p}/\lambda_{k}$, since the initial
convergence phase is not significant in the long term.) On the other
hand, obtaining an expression of $\sqrt{var(D_{k})}$ for general
topologies is difficult and deserves future research. (In Section
\ref{sub:sim_alg} we will present some numerical results.) Therefore,
how to choose $r_{max}$ to ensure that $\sqrt{var(D_{k})}$ is lower
than some threshold remains an open problem.

Next we consider the first relationship. We present several generic
bounds to characterize how the regions ${\cal C}(r_{min},r_{max})$
and ${\cal C}'(r_{min},r_{max},\Delta)$ depend on $r_{max}$ and
$r_{min}$. Given a ${\bf \lambda}\in{\cal C}$, by the definition
of ${\cal C}(r_{min},r_{max})$ in (\ref{eq:C_min_max}), if one chooses
$r_{min}<\min_{k}r_{k}^{*}({\bf \lambda})$ and $r_{max}>\max_{k}r_{k}^{*}({\bf \lambda})$,
then ${\bf \lambda}\in{\cal C}(r_{min},r_{max})$, so that Algorithm
1 can be used to support ${\bf \lambda}$. ($r_{k}^{*}({\bf \lambda})$
is the $k$-th element of ${\bf r}^{*}({\bf \lambda})$.) A similar
statement can be made for ${\cal C}'(r_{min},r_{max},\Delta)$. 

Consider a vector $\bar{{\bf \lambda}}\succ{\bf 0}$ which is at the
boundary of $\bar{{\cal C}}$ (i.e., $\bar{{\bf \lambda}}\in\bar{{\cal C}}$
but $\rho\bar{{\bf \lambda}}\notin\bar{{\cal C}},\forall\rho>1$).
Clearly, for $\rho\in(0,1)$, $\rho\bar{{\bf \lambda}}\in{\cal C}$.
Denote $\rho=1-\epsilon$. We are interested to bound ${\bf r}^{*}((1-\epsilon)\bar{{\bf \lambda}})$.
For the idealized CSMA model without collisions used in \cite{Allerton},
an earlier bound obtained in \cite{joint_work} (Lemma 8-(3)) suggests
that $\max_{k}r_{k}^{*}((1-\epsilon)\bar{{\bf \lambda}})\le O(1/\epsilon)$
(where ${\bf r}^{*}$ there controls the backoff times). In this section,
we show a stronger result that, in our model with collisions, $\max_{k}r_{k}^{*}((1-\epsilon)\bar{{\bf \lambda}})\le O(\log(1/\epsilon))$,
so that the required $r_{max}$ to support arrival rates $(1-\epsilon)\bar{{\bf \lambda}}$
is not more than $O(\log(1/\epsilon))$. (Also, one can similarly
show that the same order $O(\log(1/\epsilon))$ applies to the idealized
CSMA model as well.)
\begin{thm}
\label{thm:collision}We have\begin{equation}
\bar{{\bf \lambda}}^{T}{\bf r}^{*}((1-\epsilon)\bar{{\bf \lambda}})\le b\cdot[\log(\frac{1}{\epsilon})+\log(\frac{N'}{b})+2G+1]\label{eq:r-upper-bound-2}\end{equation}
for some constants $N'$, $b$ and $G$ (defined during the proof),
if $\epsilon\le1/b$. When $\epsilon\in(1/b,1)$, \begin{equation}
\bar{{\bf \lambda}}^{T}{\bf r}^{*}((1-\epsilon)\bar{{\bf \lambda}})\le[\log(N')+2G]/\epsilon.\label{eq:r-upper-bound-3}\end{equation}

So roughly speaking, as $\epsilon\rightarrow0$, the value of ${\bf r}^{*}({\bf \lambda})$
is not more than $O(\log(1/\epsilon))$ by (\ref{eq:r-upper-bound-2}).
\end{thm}
The proof is in Appendix D.

The following is a lower bound of ${\bf r}^{*}({\bf \lambda})$.
\begin{prop}
\label{thm:r-lower-bound-2}Given any ${\bf \lambda}\in{\cal C}$,
we have\begin{equation}
r_{k}^{*}({\bf \lambda})\ge\log(\frac{\tau'}{T_{0}}\frac{\lambda_{k}}{1-\lambda_{k}}),\forall k.\label{eq:r-lower-bound-2}\end{equation}
Therefore\begin{equation}
\min_{k}r_{k}^{*}({\bf \lambda})\ge\log(\frac{\tau'}{T_{0}}\frac{\min_{k}\lambda_{k}}{1-\min_{k}\lambda_{k}}).\label{eq:r-lower-bound-3}\end{equation}
\end{prop}
\begin{proof}
Suppose that $r_{k}^{*}({\bf \lambda})<\log(\frac{\tau'}{T_{0}}\frac{\lambda_{k}}{1-\lambda_{k}})$,
then the mean payload length is $T_{0}\exp(r_{k}^{*}({\bf \lambda}))<\tau'\cdot\lambda_{k}/(1-\lambda_{k})$.
Note that the overhead of each successful transmission is $\tau'$.
So, even if link $k$ is successfully transmitting all the time, its
service rate would be strictly less than $\tau'\cdot\frac{\lambda_{k}}{1-\lambda_{k}}/(\tau'+\tau'\cdot\frac{\lambda_{k}}{1-\lambda_{k}})=\lambda_{k}$,
leading to a contradiction. 
\end{proof}
Then, we have the following result.
\begin{cor}
$\max_{k}r_{k}^{*}((1-\epsilon)\bar{{\bf \lambda}})\le O(\log(1/\epsilon))$
as $\epsilon\rightarrow0$.\end{cor}
\begin{proof}
For convenience, denote ${\bf \lambda}=(1-\epsilon)\bar{{\bf \lambda}}$.
Since we are interested in the asymptotic behavior as $\epsilon\rightarrow0$,
assume that $\epsilon\le0.5$. Denote $\bar{\lambda}_{min}:=\min_{k}\bar{\lambda}_{k}>0$.
Then, $\lambda_{k}\ge0.5\bar{\lambda}_{min},\forall k$. 

By (\ref{eq:r-lower-bound-2}), we know that $r_{k}^{*}({\bf \lambda})\ge\log(\frac{\tau'}{T_{0}})+\log(\frac{\lambda_{k}}{1-\lambda_{k}})\ge\log(\frac{\tau'}{T_{0}})+\log(\frac{0.5\bar{\lambda}_{min}}{1-0.5\bar{\lambda}_{min}}):=\underline{r},\forall k$.
Then, combined with (\ref{eq:r-upper-bound-2}), if $\epsilon\le1/b$,
we have for any $k$,\begin{eqnarray*}
{\bf r}_{k}^{*}({\bf \lambda}) & \le & \{b\cdot[\log(\frac{N'}{b\cdot\epsilon})+2G+1]-\sum_{k'\ne k}\bar{\lambda}_{k'}\cdot r_{k'}^{*}({\bf \lambda})\}/\bar{\lambda}_{k}\\
 & \le & \{b\cdot[\log(\frac{N}{b\cdot\epsilon})+2G+1]-\sum_{k'\ne k}\bar{\lambda}_{k'}\cdot\underline{r}\}/\bar{\lambda}_{k}\\
 & = & O(\log(1/\epsilon))\end{eqnarray*}
which completes the proof.%
\begin{comment}
\emph{Remark}: The value of $b$ can be computed by (\ref{eq:b}).
But it may be computationally expensive if $N'$ is large. Obtaining
a simpler expression for $b$ is interesting for future work. However,
the above result suffices to show that $\max_{k}r_{k}^{*}((1-\epsilon)\bar{{\bf \lambda}})$
is smaller than $O(\log(1/\epsilon))$.
\end{comment}
{}
\end{proof}

\section{\label{sec:Simulation-algorithm}Numerical results}

Consider the conflict graph in Fig. \ref{fig:7nodes_CG}. Let the
vector of arrival rates be ${\bf \lambda}=\rho\cdot\bar{{\bf \lambda}}$,
where $\rho\in(0,1)$ is the {}``load'', and $\bar{{\bf \lambda}}$
is a convex combination of several maximal IS: $\bar{{\bf \lambda}}=0.2*[1,0,1,0,1,0,0]+0.2*[0,1,0,0,1,0,1]+0.2*[0,0,0,1,0,1,0]+0.2*[0,1,0,0,0,1,0]+0.2*[1,0,1,0,0,1,0]=[0.4,0.4,0.4,0.2,0.4,0.6,0.2]$.
Since $\rho\in(0,1)$, ${\bf \lambda}$ is strictly feasible. Fix
the transmission probabilities as $p_{k}=1/16,\forall k$. The {}``reference
payload length'' $T_{0}=15$. %
\begin{figure}
\begin{centering}
\includegraphics[width=4cm]{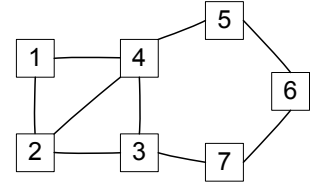}
\par\end{centering}

\caption{\label{fig:7nodes_CG}The conflict graph in simulations}

\end{figure}
%
\begin{comment}
Now we vary $\rho$ and $\eta$. And in each case we solve problem
(\ref{eq:loglikelihood-primal_c}) to obtain the required mean payload
length $T_{k}^{p}:=T_{0}\cdot\exp(r_{k}^{*}),k=1,2,\dots,7$. Fig.
\ref{fig:mean-payload} (a) shows how $T_{k}^{p}$'s change as the
load $\rho$ changes, with $\eta=1$. Clearly, as $\rho$ increases,
the $T_{k}^{p}$'s tend to increase. Also, the rate of increase becomes
faster as $\rho$ approaches 1. Therefore, as mentioned before, there
is a tradeoff between the throughput and transmission lengths (larger
transmission lengths introduce larger delays for conflicting links).
Fig. \ref{fig:mean-payload} (b) shows how the $T_{k}^{p}$'s depend
on the relative size $\eta$ of the overhead (with fixed $\rho=0.8$
and $\eta\in\{1,0.5,0.2\}$). As expected, the smaller the overhead,
the smaller $T_{k}^{p}$'s are required. 

\noindent %
\begin{figure}
\begin{centering}
\subfloat[Relation with the load (given $\eta=1$)]{\begin{centering}
\includegraphics[width=7.5cm]{length_vs_load2.eps}
\par\end{centering}

}
\par\end{centering}

\begin{centering}
\subfloat[Relation with the overhead (given $\rho=0.8$)]{\begin{centering}
\includegraphics[width=7.5cm]{length_vs_overhead2.eps}
\par\end{centering}

}
\par\end{centering}

\caption{\label{fig:mean-payload}Required mean payload lengths}

\end{figure}

\end{comment}
{}

\subsection{\label{sub:sim_alg}Transmission length control algorithms}

We evaluate algorithm (\ref{eq:Alg1b}) (a variant of Algorithm 1)
in our C++ simulator. The update in (\ref{eq:Alg1b}) is performed
every $M=500$ slots. Let the step size $\alpha(i)=0.23/(2+i/100)$,
the upper bound $r_{max}=3.5$, the lower bound $r_{min}=0$, and
the {}``gap'' $\Delta=0.005$. %
\begin{comment}
Since $T_{k}^{p}=T_{0}\exp(r_{k})$ is in general not an integer,
link $k$ chooses its payload length to be $\left\lfloor T_{k}^{p}\right\rfloor $
with probability $1-(T_{k}^{p}-\left\lfloor T_{k}^{p}\right\rfloor )$
and $\left\lfloor T_{k}^{p}\right\rfloor +1$ with probability $T_{k}^{p}-\left\lfloor T_{k}^{p}\right\rfloor $.
(Thus link $k$'s mean payload length is $T_{k}^{p}$, leading to
the desired average throughput.)
\end{comment}
{} The initial value of each $r_{k}$ is 0. 

Let the {}``load'' of arrival rates be $\rho=0.8$ (i.e., ${\bf \lambda}=0.8\cdot\bar{{\bf \lambda}}$).
The collision length (e.g., RTS length) is $\gamma=5$, and the overhead
of successful transmission is $\tau'=10$. %
\begin{comment}
It can be verified that ${\bf \lambda}\in{\cal C}'([r_{min},r_{max}],\epsilon)$. 
\end{comment}
{}To show the negative drift of the queue lengths, assume that initially
all queue lengths are 300 data units (where each data unit takes 100
slots to transmit). As expected, Fig. \ref{fig:Alg1} (a) shows the
convergence of the mean payload lengths, and Fig. \ref{fig:Alg1}
(b) shows that all queues are stable. %
\begin{figure}
\begin{centering}
\subfloat[Convergence of the mean payload lengths]{\begin{centering}
\includegraphics[width=7.5cm]{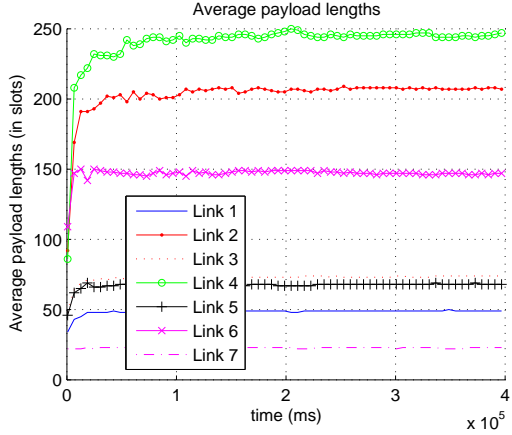}
\par\end{centering}

}
\par\end{centering}

\begin{centering}
\subfloat[Stability of the queues]{\begin{centering}
\includegraphics[width=7.5cm]{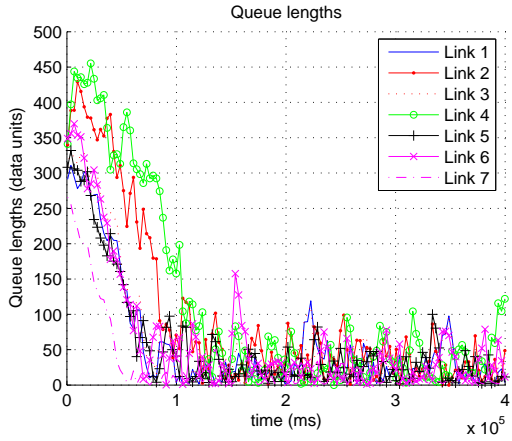}
\par\end{centering}

}
\par\end{centering}

\caption{\label{fig:Alg1}Simulation of Algorithm (\ref{eq:Alg1b}) (with the
conflict graph in Fig. \ref{fig:7nodes_CG})}

\end{figure}
\begin{table}
\begin{centering}
\begin{tabular}{|c|c|c|c|c|c|}
\hline 
$\rho$ & 0.65 & 0.7 & 0.75 & 0.8 & 0.85\tabularnewline
\hline
\hline 
Mean of $D_{3}$ (in slots) & 125.9 & 146.7 & 178.3 & 226.3 & 310.3\tabularnewline
\hline 
Standard deviation of $D_{3}$ & 132.7 & 161.9 & 211.5 & 293.7 & 456.1\tabularnewline
\hline
\end{tabular}
\par\end{centering}

\caption{\label{tab:st}Short-term fairness of link 3}

\end{table}
\begin{table}
\begin{centering}
\begin{tabular}{|c|>{\centering}p{0.2in}|>{\centering}p{0.2in}|>{\centering}p{0.2in}|>{\centering}p{0.2in}|>{\centering}p{0.2in}|>{\centering}p{0.2in}|}
\hline 
 & $R_{1}$ & $R_{2}$ & $R_{3}$ & $R_{4}$ & $R_{5}$ & $R_{6}$\tabularnewline
\hline
\hline 
$\theta=0.15$ (Simulation) & 0.279 & 0.386 & 0.547 & 0.548 & 0.387 & 0.279\tabularnewline
\hline 
$\theta=0.15$ (\cite{line_network}) & 0.272 & 0.347 & 0.442 & 0.442 & 0.347 & 0.273\tabularnewline
\hline 
$\theta=0.2$ (Simulation) & 0.526 & 0.837 & 1.372 & 1.371 & 0.840 & 0.526\tabularnewline
\hline 
$\theta=0.2$ (\cite{line_network}) & 0.5 & 0.75 & 1.125 & 1.125 & 0.75 & 0.5\tabularnewline
\hline 
$\theta=0.25$ (Simulation) & 1.075 & 2.229 & 4.735 & 4.733 & 2.240 & 1.072\tabularnewline
\hline 
$\theta=0.25$ (\cite{line_network}) & 1 & 2 & 4 & 4 & 2 & 1\tabularnewline
\hline 
$\theta=0.3$ (Simulation) & 3.210 & 12.94 & 52.76 & 52.32 & 12.91 & 3.209\tabularnewline
\hline 
$\theta=0.3$ (\cite{line_network}) & 3 & 12 & 48 & 48 & 12 & 3\tabularnewline
\hline
\end{tabular}
\par\end{centering}

\caption{\label{tab:R}Comparison of access intensities}

\end{table}

To study the tradeoff between the load $\rho$ and short-term fairness,
we run Algorithm (\ref{eq:Alg1b}) for $\rho\in\{0.65,0.7,0.75,0.8,0.85\}$.
In each case, we collect the data of the access delay and compute
its mean and standard deviation when ${\bf r}$ has almost converged.
Table \ref{tab:st} shows the results for link 3 (and other links
have a similar trend). Note that when $\rho$ increases, both the
standard deviation and the mean increase, and their ratio increases
too, indicating poorer short-term fairness.

In \cite{line_network}, van de Ven et al. considered the line topology
(i.e., 1-D lattice topology) and obtained the explicit expression
of the access intensity of each link required to support a uniform
throughput $\theta$ for all the links, under the idealized-CSMA model
without collisions \cite{csma-87,Kar,Thiran,BoE}. For comparison,
we simulate Algorithm 1 in a line topology with 6 links, where each
link conflicts with the first 2 links on both sides. After Algorithm
1 converges, we compute the access intensity of link $k$ as $R_{k}:=T_{k}^{p}/(1/p_{k}-1)$
(since the mean backoff time of link $k$ is $1/p_{k}-1$), and compare
it to the result of Theorem 2 in \cite{line_network} (although the
{}``access intensities'' under the two models are not completely
equivalent due to our inclusion of collisions.) Let $p_{k}=1/16,\forall k$,
$\gamma=1$ and $\tau'=1$. We simulate four sets of arrival rates,
${\bf \lambda}=\theta\cdot{\bf 1}$ where $\theta=0.15,0.2,0.25$
and $0.3$, and give the results in Table \ref{tab:R}. 

The results show a close match, with relatively larger differences
in $R_{3}$ and $R_{4}$. The reason is that link 3 and 4 are in the
middle of the network and suffer from more collisions. After each
collision link 3 (or 4) needs to restart the backoff, which increases
its effective backoff time and therefore requires a larger payload
length to compensate. Also, all $R_{k}$'s are higher in the simulation
due to collisions and the overhead $\tau'$.

\subsection{\label{sub:dummy_bits}Effect of dummy bits}

We have used dummy bits to facilitate our analysis and design of the
algorithms. However, transmitting dummy bits when a queue is empty
consumes extra bandwidth. In this subsection, we simulate our algorithms
without dummy bits.

In both Algorithm 1 and Algorithm (\ref{eq:Alg1b}), we make the following
heuristic modification. For each link $k$, if $\tau_{k}^{p}(i)$
as computed in (\ref{eq:randomize-Tp}) is larger than the current
(positive) queue length, then transmit a packet that includes all
the bits of the queue as the payload. That is, no dummy bits are added.
If the queue is empty then the link keeps silent. In the computation
of $s_{k}'(i)$, however, the payload of the packet is counted as
$\tau_{k}^{p}(i)$.%
\footnote{The reason for this design is that, if we only count the actual bits
transmitted, then Algorithm (\ref{eq:Alg1b}) could not converge.
Indeed, if Algorithm (\ref{eq:Alg1b}) converges, then the average
service rates is strictly larger than the arrival rates, which is
impossible if we only count the actually transmitted bits.%
} Not surprisingly, the modified algorithms are difficult to analyze,
and we therefore do not claim their convergence. (However, they still
seem to converge in the simulations.)

Fig. \ref{fig:nodummypkt} shows the evolution of the average payload
lengths under Algorithm (\ref{eq:Alg1b}) without dummy bits, when
$\rho=0.8$. Indeed, the required payload lengths are significantly
reduced compared to Fig. \ref{fig:Alg1} (a) due to the saved bandwidth. 

Under Algorithm 1, however, we find that the required payload lengths
are very close with or without dummy bits. The reason is that, since
Algorithm 1 only tries to make the average service rates equals the
arrival rates, the queues do not have a negative drift towards zero.
As a result, the queues are not close to zero most of the time, so
dummy bits are rarely generated even if they are allowed. %
\begin{figure}
\begin{centering}
\includegraphics[width=7.5cm]{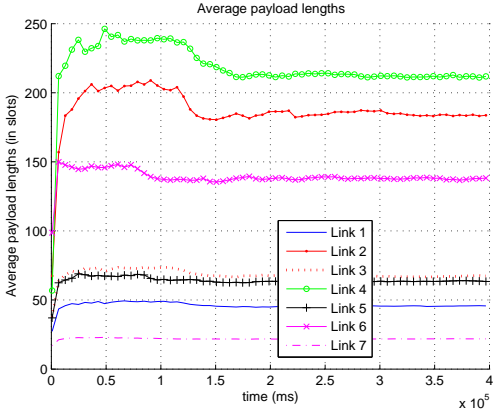}
\par\end{centering}

\caption{\label{fig:nodummypkt}Algorithm (\ref{eq:Alg1b}) without dummy bits}

\end{figure}

\subsection{\label{sub:HN}Effect of hidden nodes}

So far we have assumed that there is no hidden node (HN) in the network.
In this subsection we discuss the effect of HNs. 

Consider a simple network with 2 links that are hidden from each other.
That is, they cannot hear the transmissions of each other but a collision
occurs if their transmissions overlap. Unlike the case without HNs,
a link can start transmitting \emph{in the middle of }the other link's
transmission and cause a collision.

First, to explore how much service rates can be achieved in this scenario,
we let the two links use the same, fixed payload length $\tau^{p}$.
Let $\gamma=5,\tau'=10$, and $p_{k}=1/64,\forall k$. The two links
receive the same service rate by symmetry. Fig. \ref{fig:Service-rates-HN}
shows the service rate of one link under different values of $\tau^{p}$.
Note that the maximal service rate per link is about 0.12, much less
than 0.5 in the case without HNs. Also, when $\tau^{p}$ is large
enough, further increasing $\tau^{p}$ decreases the service rates,
because larger packets are more easily collided by the HN.

Then we simulate Algorithm 1 with arrival rates $\lambda_{1}=\lambda_{2}=0.1<0.12$.
We set $T_{0}=15,r_{min}=0,r_{max}=2.59$ so that the maximal payload
is $T_{0}\exp(r_{max})=200$ (slots), and $\alpha(i)=0.14/(2+i/100)$.
Unlike the case without HNs, the results depend on the initial condition
as shown in Fig. \ref{fig:Algorithm-1-HN}. For example, if the initial
payload lengths of both links are 40 slots (which we call {}``initial
condition 1''), then the mean payload lengths converge to the correct
value (about 17.5). However, if the initial payloads are 80 slots
({}``initial condition 2''), then the mean payload lengths keep
increasing (until reaching the maximal value) and cannot support the
arrival rates. This can be explained by Fig. \ref{fig:Service-rates-HN}.
Initial payload lengths of 40 slots achieve a per-link service rate
higher than the arrival rate. By Algorithm 1, the payload lengths
are decreased and eventually converge to the correct values. However,
if initially the payload lengths are 80 slots, a per-link service
rate lower than 0.1 is achieved. By Algorithm 1, both links increase
their payload lengths. This, however, further decreases their service
rates, and the cycle goes on. The root cause for this behavior is
as follows. Algorithm 1 has implicitly used the fact that, without
HNs, a link's service rate increases with its payload length. However,
it may not be the case when HNs exist.

To sum up, in the presence of HNs, both the achievable capacity region
of CSMA and the property of our algorithms have changed. To address
the HN problem, there are at least two directions to explore. The
first is to understand the achievable capacity with HNs, and design
algorithms to achieve the capacity. The second is to design protocols
to remove or reduce HNs, so that our existing algorithms can be applied.
There have been many proposals aiming to remove or reduce the HNs
(see \cite{Jiang-Liew} and the references therein). %
\begin{figure}
\begin{centering}
\includegraphics[width=6.5cm]{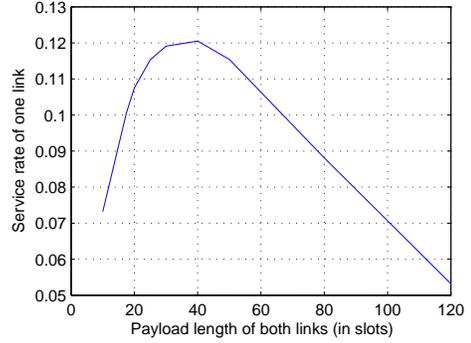}
\par\end{centering}

\caption{\label{fig:Service-rates-HN}Service rates in a 2-link network with
hidden nodes}

\end{figure}
\begin{figure}
\begin{centering}
\includegraphics[width=7.5cm]{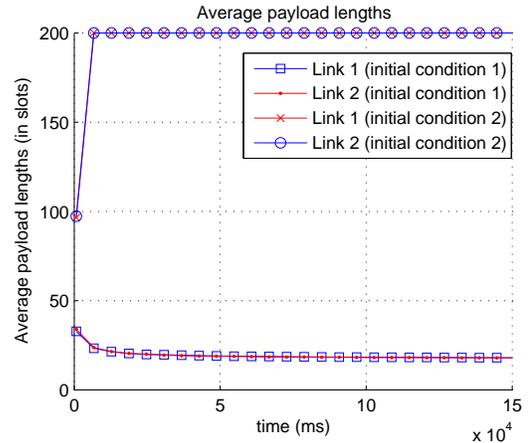}
\par\end{centering}

\caption{\label{fig:Algorithm-1-HN}Algorithm 1 with hidden nodes}

\end{figure}

\section{\label{sec:Conclusion}Conclusion}

In this paper, we have studied CSMA-based scheduling algorithms with
collisions. We first provided a model and gave a throughput formula
which takes into account the cost of collisions and overhead. The
formula has a simple product form. Next, we designed distributed algorithms
where each link adaptively updates its mean transmission length to
approach the throughput-optimality, and provided sufficient conditions
to ensure the convergence and stability of the algorithms. We also
characterized the relationship between the algorithm parameters and
the achievable capacity region. Finally, simulations results were
presented to illustrate and verify the main results.

In the algorithm, the transmission probabilities of the links are
chosen to be fixed at a reasonable level, since we have shown that
adjusting the transmission lengths alone is sufficient to approach
throughput-optimality (the main goal of this paper). However, the
choices of the transmission probabilities $p_{k}$'s has an effect
on the probability of collisions among the probe packets. In the future,
we would like to further study whether the adjustment of transmission
probabilities can be combined with the algorithm. Also, we are interested
to further explore the short-term fairness and the case with hidden
nodes.

\appendix

\section{\label{sec:Proofs-of-theorems}Proofs of theorems}

\subsection{\label{sub:Proof-product-form}Proof of Theorem \ref{thm:state-distribution}}

First, the stationary distribution of the CSMA/CA Markov chain is
expressed in the following lemma. 
\begin{lemma}
\label{thm:extended-state-distribution}In the stationary distribution,
the probability of a valid state $w$ as defined by (\ref{eq:w-definition})
is\begin{equation}
p(w)=\frac{1}{E}\prod_{i:x_{i}=0}q_{i}\prod_{j:x_{j}=1}[p_{j}\cdot f(b_{j},j,x)]\label{eq:w}\end{equation}
where \begin{equation}
f(b_{j},j,x)=\begin{cases}
1 & \text{if }j\in Z(x)\\
P_{j}(b_{j}) & \text{if }j\in S(x)\end{cases},\label{eq:f_def}\end{equation}
where $P_{j}(b_{j})$ is the p.m.f. of link $j$'s transmission length,
as defined in (\ref{eq:pmf}). Also, $E$ is a normalizing term such
that $\sum_{w}p(w)=1$, i.e., all probabilities sum up to 1. Note
that $p(w)$ does not depend on the remaining time $a_{k}$'s.\end{lemma}
\begin{proof}
For a given state $w=\{x,((b_{k},a_{k}),\forall k:x_{k}=1)\}$, define
the set of active links whose remaining time is larger than 1 as\[
A_{1}(w)=\{k|x_{k}=1,a_{k}>1\}.\]
Links in $A_{1}(w)$ will continue their transmissions (either with
success or a collision) in the next slot. 

Define the set of inactive links {}``blocked'' by links in $A_{1}(w)$
as\[
\partial A_{1}(w)=\{j|x_{j}=0;e(j,k)=1\text{ for some }k\in A_{1}(w)\}\]
where $e(j,k)=1$ means that there is an edge between $j$ and $k$
in the conflict graph. Links in $\partial A_{1}(w)$ will remain inactive
in the next slot. 

Write $\bar{A}_{1}(w):=A_{1}(w)\cup\partial A_{1}(w)$. Define the
set of all other links as\begin{equation}
A_{2}(w)={\cal N}\backslash\bar{A}_{1}(w).\label{eq:A_2}\end{equation}
These links can change their on-off states $x_{k}$'s in the next
slot. On the other hand, links in $\bar{A}_{1}(w)$ will have the
same on-off states $x_{k}$'s in the next slot.

To illustrate these notations, consider the example in Fig. \ref{fig:Example-Markov}.
By (\ref{eq:w-example}), we have $A_{1}(w)=\{3\},\partial A_{1}(w)=\{2\},\bar{A}_{1}(w)=\{2,3\}$
and $A_{2}(w)=\{1\}$.

State $w$ can transit in the next slot to another valid state $w'=\{x',((b_{k}',a_{k}'),\forall k:x_{k}'=1)\}$,
i.e., $Q(w,w')>0$, if and only if $w'$ satisfies that (i) $x_{k}'=x_{k},\forall k\in\bar{A}_{1}(w)$;
(ii) $b_{k}'=b_{k},a_{k}'=a_{k}-1,\forall k\in\bar{A}_{1}(w)\text{ such that }x_{k}=1$;
(iii) $a_{k}'=b_{k}',\forall k\in A_{2}(w)\text{ such that }x_{k}'=1$,
and $b_{k}'=\gamma,\forall k\in A_{2}(w)\cap Z(x')$. (If $A_{2}(w)$
is an empty set, then condition (iii) is trivially true.) The transition
probability is \begin{equation}
Q(w,w')=\prod_{i\in A_{2}(w)}[p_{i}\cdot f(b_{i}',i,x')]^{x_{i}'}q_{i}^{1-x_{i}'}.\label{eq:transition}\end{equation}
Define\begin{equation}
\tilde{Q}(w',w):=\prod_{i\in A_{2}(w)}[p_{i}\cdot f(b_{i},i,x)]^{x_{i}}q_{i}^{1-x_{i}}.\label{eq:transition-reversed}\end{equation}
(If $A_{2}(w)$ is an empty set, then $Q(w,w')=1$ and $\tilde{Q}(w',w):=1$.)
If $w$ and $w'$ does not satisfy all the conditions (i), (ii), (iii),
then $Q(w,w')=0$, and also define $\tilde{Q}(w',w)=0$. 

Then, if $Q(w,w')>0$ (and $\tilde{Q}(w',w)>0$), $p(w)/\tilde{Q}(w',w)=\frac{1}{E}\prod_{i\notin A_{2}(w)}[p_{i}\cdot f(b_{i},i,x)]^{x_{i}}q_{i}^{1-x_{i}}$.
And $p(w')/Q(w,w')=\frac{1}{E}\prod_{i\notin A_{2}(w)}[p_{i}\cdot f(b_{i}',i,x')]^{x_{i}'}q_{i}^{1-x_{i}'}$.
But for any $i\notin A_{2}(w)$, i.e., $i\in\bar{A}_{1}(w)$, we have
$x_{i}'=x_{i},b_{i}'=b_{i}$ by condition (i), (ii) above. Therefore,
the two expressions are equal. Thus\begin{equation}
p(w)Q(w,w')=p(w')\tilde{Q}(w',w).\label{eq:balance}\end{equation}

If two states $w,w'$ satisfy $Q(w,w')=0$, then by definition $\tilde{Q}(w',w)=0$,
making (\ref{eq:balance}) trivially true. Therefore, (\ref{eq:balance})
holds for any $w$ and $w'$.

We will show later that $\tilde{Q}(w',w)$ is the transition probability
of the {}``time-reversed process'' of the above Markov chain (notice
the similarity between $Q(w,w')$ and $\tilde{Q}(w',w)$), and naturally
satisfies $\sum_{w}\tilde{Q}(w',w)=1$. Assuming that the claim is
true, then by (\ref{eq:balance}), we have \[
\sum_{w}p(w)\cdot Q(w,w')=p(w')\sum_{w}\tilde{Q}(w',w)=p(w').\]
Therefore, $p(w)$ is the stationary (or {}``invariant'') distribution,
which completes the proof of Lemma \ref{thm:extended-state-distribution}.

It remains to be shown that the above claim is true, in particular,
that $\sum_{w}\tilde{Q}(w',w)=1$. Denote the orignal process as $\{w(t),t\in{\cal Z}\}$
(this is the Markov process that describes our CSMA protocol, with
transition probabilities $Q(\cdot,\cdot)$ in (\ref{eq:transition})).
Adding a time index in (\ref{eq:w-definition}), we have \begin{equation}
w(t):=\{x(t),((b_{k}(t),a_{k}(t)),\forall k:x_{k}(t)=1)\}.\label{eq:w-time}\end{equation}
Now define the time-reversed process $\tilde{w}(t):=w(-t),\forall t\in{\cal Z}$.
First, note that in the process $\{w(t)\}$, the remaining time $a_{k}(t)$,
if defined, decreases with $t$; in the reversed process $\{\tilde{w}(t)\}$,
however, $a_{k}(-t)$ increases with $t$. Therefore, $\{w(t)\}$
and $\{\tilde{w}(t)\}$ are, clearly, statistically distinguishable.
So $\{w(t)\}$ is not time-reversible.

However, if we re-label the {}``remaining time'' in the reversed
order, then the process $\{\tilde{w}(t)\}$ {}``looks a lot'' like
the process $\{w(t)\}$. (This is why we say the CSMA/CA Markov chain
is almost time-reversible.) More formally, with the understanding
that $w=\{x,((b_{k},a_{k}),\forall k:x_{k}=1)\}$, define a function
$g(\cdot)$ as \begin{equation}
g(w)=\{x,((b_{k},b_{k}-a_{k}+1),\forall k:x_{k}=1)\}.\label{eq:func_g}\end{equation}

Then define the process \[
\hat{w}(t):=g(w(-t)).\]
Note that in the process $\{\hat{w}(t)\}$, the {}``remaining time''
decreases with $t$, similar to $\{w(t)\}$. 

Next we show the following two facts.

Fact 1: For any two states $w$ and $w'$ with $Q(w,w')>0$ (i.e.,
if the CSMA Markov chain can transit from state $w$ to state $w'$),
we have $A_{1}(w)=A_{1}(g(w'))$, $\bar{A}_{1}(w)=\bar{A}_{1}(g(w'))$
and $A_{2}(w)=A_{2}(g(w'))$. 

Fact 2: $Q(w,w')>0\Leftrightarrow Q(g(w'),g(w))>0$.

These facts can be illustrated by the example in Fig. \ref{fig:Example-Markov}.
First consider Fact 1. Note that $A_{1}(w)=\{3\}$, by definition,
is the set of links that are in the middle of a transmission in state
$w$ and will continue the transmission in the next state $w'$. Then,
in the reversed process, such links are also in the middle of a transmission
in state $g(w')$ and will continue the transmission in the next state
$g(w)$. So $A_{1}(w)=A_{1}(g(w'))$. Similarly, $\partial A_{1}(w)=\{2\}$,
the set of links that are blocked by $A_{1}(w)$ in $w$ are also
blocked by $A_{1}(g(w'))$ in the reversed process. Therefore $\partial A_{1}(w)=\partial A_{1}(g(w'))$.
Then by (\ref{eq:A_2}), we have $A_{2}(w)=A_{2}(g(w'))$. (Note that
it is not difficult to prove Fact 1 mechanically via the definitions
of $A_{1}(\cdot),A_{2}(\cdot)$. But we omit it here.) 

One can also verify Fact 2 in Fig. \ref{fig:Example-Markov}. We now
give a more formal proof. If $Q(w,w')>0$, then $w$ and $w'$ satisfy
conditions (i)\textasciitilde{}(iii). We first show that $Q(g(w'),g(w))>0$.
To this end, we need to verify that the states $g(w')$ and $g(w)$
satisfy condition (i)\textasciitilde{}(iii). Condition (i) holds because
$\bar{A}_{1}(w)=\bar{A}_{1}(g(w'))$ by Fact 1, and because $g(\cdot)$
does not change the {}``on-off state'' of its argument. Condition
(ii) holds since $g(\cdot)$ has reversed the remaining time (cf.
(\ref{eq:func_g})). Condition (iii) requires that in the reversed
process, any link $k\in A_{2}(g(w'))$ which is transmitting in the
state $g(w)$ must have just started its transmission. This is true
because $A_{2}(g(w'))=A_{2}(w)$ by Fact 1, and that in the original
process $w(t)$, any link $k\in A_{2}(w)$ which is transmitting in
state $w$ must be in its last slot of the transmission (otherwise
the link would be in $A_{1}(w)$). Then condition (iii) holds since
$g(\cdot)$ has reversed the remaining time.

This completes the proof that $Q(w,w')>0\Rightarrow Q(g(w'),g(w))>0$.
Now, if $Q(g(w'),g(w))>0$, by the above result, we have $Q(g(g(w)),g(g(w')))>0$.
Since $g(g(w))=w,g(g(w'))=w'$, we have $Q(w,w')>0$. This completes
the proof of Fact 2.

Consider two states $w$ and $w'$ with $Q(w,w')>0$. Then $Q(g(w'),g(w))>0$,
with \[
Q(g(w'),g(w))=\prod_{i\in A_{2}(g(w'))}[p_{i}\cdot f(b_{i},i,x)]^{x_{i}}q_{i}^{1-x_{i}}\]
by (\ref{eq:transition}). Using (\ref{eq:transition-reversed}) and
$A_{2}(w)=A_{2}(g(w'))$, we have \begin{eqnarray}
\tilde{Q}(w',w) & = & \prod_{i\in A_{2}(w)}[p_{i}\cdot f(b_{i},i,x)]^{x_{i}}q_{i}^{1-x_{i}}\nonumber \\
 & = & \prod_{i\in A_{2}(g(w'))}[p_{i}\cdot f(b_{i},i,x)]^{x_{i}}q_{i}^{1-x_{i}}\nonumber \\
 & = & Q(g(w'),g(w)).\label{eq:prob-equal}\end{eqnarray}

Therefore, $\tilde{Q}(w',w)$ is the transition probability of the
reversed process. 

By definition, $\tilde{Q}(w',w)=0$ for any $w,w'$ satisfying $Q(w,w')=0$.
So, given $w'$,\begin{eqnarray*}
\sum_{w}\tilde{Q}(w',w) & = & \sum_{w:Q(w,w')>0}\tilde{Q}(w',w)\\
 & = & \sum_{w:Q(w,w')>0}Q(g(w'),g(w))\\
 & = & \sum_{w:Q(g(w'),g(w))>0}Q(g(w'),g(w))\\
 & = & \sum_{w}Q(g(w'),g(w))=1\end{eqnarray*}
where the last step has used the fact that $g(\cdot)$ is a one-one
mapping, so that the summation is over all valid states.
\end{proof}
\medskip{}
Using Lemma \ref{thm:extended-state-distribution}, the probability
of any on-off state $x$, as in Theorem \ref{thm:state-distribution},
can be computed by summing up the probabilities of all states $w$'s
with the same on-off state $x$, using (\ref{eq:w}). 

Define the set of valid states ${\cal B}(x):=\{w|\text{ the on-off state is }x\text{ in the state }w\}$.
By Lemma \ref{thm:extended-state-distribution}, we have\begin{eqnarray}
 &  & p(x)=\sum_{w\in{\cal B}(x)}p(w)\nonumber \\
 & = & \frac{1}{E}\sum_{w\in{\cal B}(x)}\{\prod_{i:x_{i}=0}q_{i}\prod_{j:x_{j}=1}[p_{j}\cdot f(b_{j},j,x)]\}\nonumber \\
 & = & \frac{1}{E}(\prod_{i:x_{i}=0}q_{i}\prod_{j:x_{j}=1}p_{j})\sum_{w\in{\cal B}(x)}\prod_{j:x_{j}=1}f(b_{j},j,x)\nonumber \\
 & = & \frac{1}{E}(\prod_{i:x_{i}=0}q_{i}\prod_{j:x_{j}=1}p_{j})\cdot\sum_{w\in{\cal B}(x)}[\prod_{j\in S(x)}P_{j}(b_{j})].\label{eq:step1}\end{eqnarray}

Now we compute the term $\sum_{w\in{\cal B}(x)}[\prod_{j\in S(x)}P_{j}(b_{j})]$.
Consider a state $w=\{x,((b_{k},a_{k}),\forall k:x_{k}=1)\}\in{\cal B}(x)$.
For $k\in S(x)$, $b_{k}$ can take different values in ${\cal Z}_{++}$.
For each fixed $b_{k}$, $a_{k}$ can be any integer from 1 to $b_{k}$.
For a collision component $C_{m}(x)$ (i.e., $|C_{m}(x)|>1$), the
remaining time of each link in the component, $a^{(m)}$, can be any
integer from 1 to $\gamma$. Then we have \begin{eqnarray}
 &  & \sum_{w\in{\cal B}(x)}[\prod_{j\in S(x)}P_{j}(b_{j})]\nonumber \\
 & = & \prod_{j\in S(x)}[\sum_{b_{j}}\sum_{1\le a_{j}\le b_{j}}P_{j}(b_{j})]\prod_{m:|C_{m}(x)|>1}(\sum_{1\le a^{(m)}\le\gamma}1)\nonumber \\
 & = & \prod_{j\in S(x)}[\sum_{b_{j}}b_{j}P_{j}(b_{j})]\cdot\gamma^{h(x)}\nonumber \\
 & = & (\prod_{j\in S(x)}T_{j})\gamma^{h(x)}.\label{eq:step2}\end{eqnarray}

Combining (\ref{eq:step1}) and (\ref{eq:step2}) completes the proof.

\subsection{\label{sub:Proof-of-rstar}Proof of Theorem \ref{thm:thu-optimal}}

\subsubsection{Some definitions}

If at an on-off state $x$, $k\in S(x)$ (i.e., $k$ is transmitting
successfully), it is possible that link $k$ is transmitting the overhead
or the payload. So we define the {}``detailed state'' $(x,z)$,
where $z\in\{0,1\}^{K}$. Let $z_{k}=1$ if $k\in S(x)$ and link
$k$ is transmitting its payload (instead of overhead). Let $z_{k}=0$
otherwise. Denote the set of all possible detailed states $(x,z)$
by ${\cal S}$. 

Then similar to the proof of Theorem \ref{thm:state-distribution},
and using equation (\ref{eq:px_given_r}), we have the following product-form
stationary distribution\begin{equation}
p((x,z);{\bf r})=\frac{1}{E({\bf r})}g(x,z)\cdot\exp(\sum_{k}z_{k}r_{k})\label{eq:detailed-state-dist}\end{equation}
where \begin{equation}
g(x,z)=g(x)\cdot(\tau')^{|S(x)|-{\bf 1}'{\bf z}}T_{0}^{{\bf 1}'{\bf z}}\label{eq:g_x_z}\end{equation}
where ${\bf 1}'{\bf z}$ is the number of links that are transmitting
the payload in state $(x,z)$.

Clearly, this provides another expression of the service rate $s_{k}({\bf r})$:\begin{equation}
s_{k}({\bf r})=\sum_{(x,z)\in{\cal S}:z_{k}=1}p((x,z);{\bf r}).\label{eq:detailed-state-s_k}\end{equation}
\medskip{}

Now we give alternative definitions of feasible and strictly feasible
arrival rates to facilitate our proof. We will show that these definitions
are equivalent to Definition \ref{def:feasible}.
\begin{definitn}
\label{def:C-Cbar-CO}(i) A vector of arrival rate ${\bf \lambda}\in{\cal R}_{+}^{K}$
(where $K$ is the number of links) is \emph{feasible} if there exists
a probability distribution $\bar{{\bf p}}$ over ${\cal S}$ (i.e.,
$\sum_{(x,z)\in{\cal S}}\bar{p}((x,z))=1$ and $\bar{p}((x,z))\ge0$),
such that \begin{equation}
\lambda_{k}=\sum_{(x,z)\in{\cal S}}\bar{p}((x,z))\cdot z_{k}.\label{eq:feasible}\end{equation}

Let $\bar{{\cal C}}_{CO}$ be the set of feasible ${\bf \lambda}$,
where {}``CO'' stands for {}``collision''.

The rationale of the definition is that if ${\bf \lambda}$ can be
scheduled by the network, the fraction of time that the network spent
in the detailed states must be non-negative and sum up to 1. (Note
that (\ref{eq:feasible}) is the probability that link $k$ is sending
its payload given the distribution of the detailed states.)

(ii) A vector of arrival rate ${\bf \lambda}\in{\cal R}_{+}^{K}$
is \emph{strictly feasible} if it can be written as (\ref{eq:feasible})
where $\sum_{(x,z)\in{\cal S}}\bar{p}((x,z))=1$ and $\bar{p}((x,z))>0$.
Let ${\cal C}_{CO}$ be the set of strictly feasible ${\bf \lambda}$. 

For example, in the ad-hoc network in Fig. \ref{fig:Ad-hoc} (b),
${\bf \lambda}=(0.5,0.5,0.5)$ is feasible, because (\ref{eq:feasible})
holds if we let the probability of the detailed state $(x=(1,0,1),z=(1,0,1))$
be 0.5, the probability of the detailed state $(x=(0,1,0),z=(0,1,0))$
be 0.5, and all other detailed states have probability 0. However,
${\bf \lambda}=(0.5,0.5,0.5)$ is not strictly feasible since it cannot
be written as (\ref{eq:feasible}) where all $\bar{p}((x,z))>0$.
But ${\bf \lambda}'=(0.49,0.49,0.49)$ is strictly feasible.\medskip{}
\end{definitn}
\begin{prop}
The above definitions are equivalent to Definition \ref{def:feasible}.
That is,\begin{eqnarray}
\bar{{\cal C}}_{CO} & = & \bar{{\cal C}}\label{eq:eq_Cbar}\\
{\cal C}_{CO} & = & {\cal C}.\label{eq:eq_C}\end{eqnarray}
\end{prop}
\begin{proof}
We first prove (\ref{eq:eq_Cbar}). By definition, any ${\bf \lambda}\in\bar{{\cal C}}$
can be written as ${\bf \lambda}=\sum_{\sigma\in{\cal X}}\bar{p}_{\sigma}\sigma$
where ${\cal X}$ is the set of independent sets, and $\bar{{\bf p}}=(\bar{p}_{\sigma})_{\sigma\in{\cal X}}$
is a probability distribution, i.e., $\bar{p}_{\sigma}\ge0,\sum_{\sigma\in{\cal X}}\bar{p}_{\sigma}=1$.
Now we construct a distribution ${\bf p}$ over the states $(x,z)\in{\cal S}$
as follows. Let $p((\sigma,\sigma))=\bar{p}_{\sigma},\forall\sigma\in{\cal X}$,
and let $p((x,z))=0$ for all other states $(x,z)\in{\cal S}$. Then,
clearly $\sum_{(x,z)\in{\cal S}}p((x,z))\cdot z=\sum_{\sigma\in{\cal X}}p((\sigma,\sigma))\cdot\sigma=\sum_{\sigma\in{\cal X}}\bar{p}_{\sigma}\sigma={\bf \lambda}$,
which implies that ${\bf \lambda}\in\bar{{\cal C}}_{CO}$. So, \begin{equation}
\bar{{\cal C}}\subseteq\bar{{\cal C}}_{CO}.\label{eq:subset1}\end{equation}

On the other hand, if ${\bf \lambda}\in\bar{{\cal C}}_{CO}$, then
${\bf \lambda}=\sum_{(x,z)\in{\cal S}}p((x,z))\cdot z$ for some distribution
${\bf p}$ over ${\cal S}$. We define another distribution $\bar{{\bf p}}$
over ${\cal X}$ as follows. Let $\bar{p}_{\sigma}=\sum_{(x,z)\in{\cal S}:z=\sigma}p((x,z)),\forall\sigma\in{\cal X}$.
Then, ${\bf \lambda}=\sum_{(x,z)\in{\cal S}}p((x,z))\cdot z=\sum_{\sigma\in{\cal X}}\sum_{(x,z)\in{\cal S}:z=\sigma}p((x,z))\sigma=\sum_{\sigma\in{\cal X}}\bar{p}_{\sigma}\sigma$,
which implies that ${\bf \lambda}\in\bar{{\cal C}}$. Therefore\begin{equation}
\bar{{\cal C}}_{CO}\subseteq\bar{{\cal C}}.\label{eq:subset2}\end{equation}

Combining (\ref{eq:subset1}) and (\ref{eq:subset2}) yields (\ref{eq:eq_Cbar}).

We defined that ${\cal C}$ is the interior of $\bar{{\cal C}}$.
To prove (\ref{eq:eq_C}), we only need to show that ${\cal C}_{CO}$
is also the interior of $\bar{{\cal C}}$. The proof is similar to
that in Appendix A of \cite{CSMA_longer}, and is thus omitted.
\end{proof}

\subsubsection{\label{sub:attain}Existence of ${\bf r}^{*}$ }

Assume that ${\bf \lambda}$ is strictly feasible. Consider the following
convex optimization problem, where the vector ${\bf u}$ can be viewed
as a probability distribution over the detailed states $(x,z)$:\begin{eqnarray}
 & \max_{{\bf u}} & \{H({\bf u})+\sum_{(x,z)\in{\cal S}}[u_{(x,z)}\cdot\log(g(x,z))]\}\nonumber \\
 & \text{s.t.} & \sum_{(x,z)\in{\cal S}:z_{k}=1}u_{(x,z)}=\lambda_{k},\forall k\nonumber \\
 &  & u_{(x,z)}\ge0,\sum_{(x,z)\in{\cal S}}u_{(x,z)}=1\label{eq:ME-dual}\end{eqnarray}
where $H({\bf u}):=\sum_{(x,z)\in{\cal S}}[-u_{(x,z)}\log(u_{(x,z)})]$
is the {}``entropy'' of the distribution ${\bf u}$.

Let $r_{k}$ be the dual variable associated with the constraint $\sum_{(x,z)\in{\cal S}:z_{k}=1}u_{(x,z)}=\lambda_{k}$,
and let the vector ${\bf r}:=(r_{k})$. We will show the following. 
\begin{lemma}
The optimum dual variables ${\bf r}^{*}$ (when problem (\ref{eq:ME-dual})
is solved) exists, and satisfy (\ref{eq:thu-optimal}), i.e., $s_{k}({\bf r}^{*})=\lambda_{k},\forall k$.
Also, the dual problem of (\ref{eq:ME-dual}) is (\ref{eq:loglikelihood-primal_c}).\end{lemma}
\begin{proof}
With the above definition of ${\bf r}$, a partial Lagrangian of problem
(\ref{eq:ME-dual}) (subject to $u_{(x,z)}\ge0,\sum_{(x,z)}u_{(x,z)\in{\cal S}}=1$)
is\begin{eqnarray}
 &  & {\cal L}({\bf u};{\bf r})\nonumber \\
 & = & H({\bf u})+\sum_{(x,z)\in{\cal S}}[u_{(x,z)}\log(g(x,z))]\nonumber \\
 &  & +\sum_{k}r_{k}[\sum_{(x,z)\in{\cal S}:z_{k}=1}u_{(x,z)}-\lambda_{k}]\nonumber \\
 & = & \sum_{(x,z)\in{\cal S}}\{u_{(x,z)}[-\log(u_{(x,z)})+\log(g(x,z))\nonumber \\
 &  & +\sum_{k:z_{k}=1}r_{k}]\}-\sum_{k}(r_{k}\lambda_{k}).\label{eq:partial-L}\end{eqnarray}

So\[
\frac{\partial{\cal L}({\bf u};{\bf r})}{\partial u_{(x,z)}}=-\log(u_{(x,z)})-1+\log(g(x,z))+\sum_{k:z_{k}=1}r_{k}.\]

We claim that \begin{equation}
u_{(x,z)}({\bf r}):=p((x,z);{\bf r}),\forall(x,z)\in{\cal S}\label{eq:Lagrangian-maximizer}\end{equation}
 (cf. equation (\ref{eq:detailed-state-dist})) maximizes ${\cal L}({\bf u};{\bf r})$
over ${\bf u}$ subject to $u_{(x,z)}\ge0,\sum_{(x,z)}u_{(x,z)\in{\cal S}}=1$.
Indeed, the partial derivative at the point ${\bf u}({\bf r})$ is\[
\frac{\partial{\cal L}({\bf u}({\bf r});{\bf r})}{\partial u_{(x,z)}}=\log(E({\bf r}))-1,\]
which is the same for all $(x,z)\in{\cal S}$ (Since given the dual
variables ${\bf r}$, $\log(E({\bf r}))$ is a constant). Also, $u_{(x,z)}({\bf r})=p((x,z);{\bf r})>0$
and $\sum_{(x,z)\in{\cal S}}u_{(x,z)}({\bf r})=1$. Therefore, it
is impossible to increase ${\cal L}({\bf u};{\bf r})$ by slightly
perturbing ${\bf u}$ around ${\bf u}({\bf r})$ (subject to ${\bf 1}^{T}{\bf u}=1$).
Since ${\cal L}({\bf u};{\bf r})$ is concave in ${\bf u}$, the claim
follows. 

Denote $l({\bf r})=\max_{{\bf u}}{\cal L}({\bf u};{\bf r})$, then
the dual problem of (\ref{eq:ME-dual}) is $\inf_{{\bf r}}l({\bf r})$.
Plugging the expression of $u_{(x,z)}({\bf r})$ into ${\cal L}({\bf u};{\bf r})$,
it is not difficult to find that $\inf_{{\bf r}}l({\bf r})$ is equivalent
to $\sup_{{\bf r}}L({\bf r};{\bf \lambda})$ where $L({\bf r};{\bf \lambda})$
is defined in (\ref{eq:Lr}).

Since ${\bf \lambda}$ is strictly feasible, it can be written as
(\ref{eq:feasible}) where $\sum_{(x,z)\in{\cal S}}\bar{p}((x,z))=1$
and $\bar{p}((x,z))>0$. Therefore, there exists ${\bf u}\succ{\bf 0}$
(by choosing ${\bf u}=\bar{{\bf p}}$) that satisfies the constraints
in (\ref{eq:ME-dual}) and also in the interior of the domain of the
objective function. So, problem (\ref{eq:ME-dual}) satisfies the
Slater condition \cite{convex-book}. As a result, there exists a
vector of (finite) optimal dual variables ${\bf r}^{*}$ when problem
(\ref{eq:ME-dual}) is solved. Also, ${\bf r}^{*}$ solves the dual
problem $\sup_{{\bf r}}L({\bf r};{\bf \lambda})$. Therefore, $\sup_{{\bf r}}L({\bf r};{\bf \lambda})$
is attainable and can be written as $\max_{{\bf r}}L({\bf r};{\bf \lambda})$,
as in (\ref{eq:loglikelihood-primal_c}).

Finally, the optimal solution ${\bf u}^{*}$ of problem (\ref{eq:ME-dual})
is such that $u_{(x,z)}^{*}=u_{(x,z)}({\bf r}^{*}),\forall(x,z)\in{\cal S}$.
Also, ${\bf u}^{*}$ is clearly feasible for problem (\ref{eq:ME-dual}).
Therefore, \[
\sum_{(x,z)\in{\cal S}:z_{k}=1}u_{(x,z)}^{*}=s_{k}({\bf r}^{*})=\lambda_{k},\forall k.\]

\end{proof}
\emph{Remark}: From (\ref{eq:partial-L}) and (\ref{eq:Lagrangian-maximizer}),
we see that a subgradient (or gradient) of the dual objective function
$L({\bf r};{\bf \lambda})$ is \[
\frac{\partial L({\bf r};{\bf \lambda})}{\partial r_{k}}=\lambda_{k}-\sum_{(x,z)\in{\cal S}:z_{k}=1}u_{(x,z)}({\bf r})=\lambda_{k}-s_{k}({\bf r}).\]
This can also be obtained by direct differentiation of $L({\bf r};{\bf \lambda})$.

\subsubsection{\label{sub:uniqueness}Uniqueness of ${\bf r}^{*}$}

Now we show the uniqueness of ${\bf r}^{*}$. Note that the objective
function of (\ref{eq:ME-dual}) is strictly concave. Therefore ${\bf u}^{*}$,
the optimal solution of (\ref{eq:ME-dual}) is unique. Consider two
detailed state $({\bf e}_{k},{\bf e}_{k})$ and $({\bf e}_{k},{\bf 0})$,
where ${\bf e}_{k}$ is the $K$-dimensional vector whose $k$'th
element is 1 and all other elements are 0's. We have $u_{({\bf e}_{k},{\bf e}_{k})}^{*}=p(({\bf e}_{k},{\bf e}_{k});{\bf r}^{*})$
and $u_{({\bf e}_{k},{\bf 0})}^{*}=p(({\bf e}_{k},{\bf 0});{\bf r}^{*})$.
Then by (\ref{eq:detailed-state-dist}), \begin{equation}
u_{({\bf e}_{k},{\bf e}_{k})}({\bf r}^{*})/u_{({\bf e}_{k},{\bf 0})}({\bf r}^{*})=\exp(r_{k}^{*})\cdot(T_{0}/\tau').\label{eq:ratio-prob}\end{equation}
Suppose that $r^{*}$ is not unique, that is, there exist ${\bf r}_{I}^{*}\ne{\bf r}_{II}^{*}$
but both are optimal ${\bf r}$. Then, $r_{I,k}^{*}\ne r_{II,k}^{*}$
for some $k$. This contradicts (\ref{eq:ratio-prob}) and the uniqueness
of ${\bf u}^{*}$. Therefore ${\bf r}^{*}$ is unique. This also implies
that $\max_{{\bf r}}L({\bf r};{\bf \lambda})$ has a unique solution
${\bf r}^{*}$.

\subsection{\label{sub:Proof-of-Theorem}Proof of Theorem \ref{thm:scheduling}}

We will use results in \cite{Borkar} to prove Theorem \ref{thm:scheduling}.
Similar techniques have been used in \cite{conv_MS} to analyze the
convergence of an algorithm in \cite{Allerton}.

\subsubsection{Part (i): Decreasing step size}

Define the concave function

\begin{equation}
\tilde{H}(y):=\begin{cases}
-(r_{min}-y)^{2}/2 & \text{if }y<r_{min}\\
0 & \text{if }y\in[r_{min},r_{max}]\\
-(r_{max}-y)^{2}/2 & \text{if }y>r_{max}.\end{cases}\label{eq:H}\end{equation}

Note that $d\tilde{H}(y)/dy=\tilde{h}(y)$ where $\tilde{h}(y)$ is
defined in (\ref{eq:h}). Let $G({\bf r};{\bf \lambda}):=L({\bf r};{\bf \lambda})+\sum_{k}\tilde{H}(r_{k})$.
Since ${\bf \lambda}$ is strictly feasible, $\max_{{\bf r}}L({\bf r};{\bf \lambda})$
has a unique solution ${\bf r}^{*}$. That is, $L({\bf r}^{*};{\bf \lambda})>L({\bf r};{\bf \lambda}),\forall{\bf r}\ne{\bf r}^{*}$.
Since ${\bf r}^{*}\in(r_{min},r_{max})^{K}$ by assumption, then $\forall{\bf r}$,
$\sum_{k}\tilde{H}(r_{k}^{*})=0\ge\sum_{k}\tilde{H}(r_{k})$. Therefore,
$G({\bf r}^{*};{\bf \lambda})>G({\bf r};{\bf \lambda}),\forall{\bf r}\ne{\bf r}^{*}$.
So ${\bf r}^{*}$ is the unique solution of $\max_{{\bf r}}G({\bf r};{\bf \lambda})$.
Because $\partial G({\bf r};{\bf \lambda})/\partial r_{k}=\lambda_{k}-s_{k}({\bf r})+\tilde{h}(r_{k})$,
Algorithm 1 tries to solve $\max_{{\bf r}}G({\bf r};{\bf \lambda})$
with inaccurate gradients. 

Let ${\bf v}^{s}(t)$ be the solution of the following differential
equation (for $t\ge s$) \begin{equation}
dv_{k}(t)/dt=\lambda_{k}-s_{k}({\bf v}(t))+\tilde{h}(v_{k}(t)),\forall k\label{eq:ODE}\end{equation}
with the initial condition that ${\bf v}^{s}(s)=\bar{{\bf r}}(s)$.
So, ${\bf v}^{s}(t)$ can be viewed as the {}``ideal'' trajectory
of Algorithm 1 with the smoothed arrival rate and service rate. And
(\ref{eq:ODE}) can be viewed as a continuous-time gradient algorithm
to solve $\max_{{\bf r}}G({\bf r};{\bf \lambda})$. We have shown
above that ${\bf r}^{*}$ is the unique solution of $\max_{{\bf r}}G({\bf r};{\bf \lambda})$.
Therefore ${\bf v}^{s}(t)$ converges to the unique ${\bf r}^{*}$
for any initial condition.

Recall that in Algorithm 1, ${\bf r}(i)$ is always updated at the
beginning of a minislot. Define $Y(i-1):=(s_{k}'(i),w_{0}(i))$ where
$w_{0}(i)$ is the state $w$ at time $t_{i}$. Then $\{Y(i)\}$ is
a non-homogeneous Markov process whose transition kernel from time
$t_{i-1}$ to $t_{i}$ depends on ${\bf r}(i-1)$. The update in Algorithm
1 can be written as\[
r_{k}(i)=r_{k}(i-1)+\alpha(i)\cdot[f(r_{k}(i-1),Y(i-1))+M(i)]\]
where $f(r_{k}(i-1),Y(i-1)):=\lambda_{k}-s_{k}'(i)+\tilde{h}(r_{k}(i-1))$,
and $M(i)=\lambda_{k}'(i)-\lambda_{k}$ is zero-mean noise.

To use Corollary 8 in page 74 of \cite{Borkar} to show Algorithm
1's almost-sure convergence to ${\bf r}^{*}$, the following conditions
are sufficient: 

(i) $f(\cdot,\cdot)$ is Lipschitz in the first argument, and uniformly
in the second argument. This holds by the construction of $\tilde{h}(\cdot)$; 

(ii) The transition kernel of $Y(i)$ is continuous in ${\bf r}(i)$.
This is true due to the way we randomize the transmission lengths
in (\ref{eq:randomize-Tp}).

(iii) (\ref{eq:ODE}) has a unique convergent point ${\bf r}^{*}$,
which has been shown above; 

(iv) With Algorithm 1, $r_{k}(i)$ is bounded $\forall k,i$ almost
surely. This is proved in Lemma \ref{lem:r-bounded-Algorithm-4} below.

(v) Tightness condition ((\dag{}) in \cite{Borkar}, page 71): This
is satisfied since $Y(i)$ has a bounded state-space (cf. conditions
(6.4.1) and (6.4.2) in \cite{Borkar}, page 76). The state space of
$Y(i)$ is bounded because $s_{k}'(i)\in[0,1]$ and $w_{0}(i)$ is
in a finite set (which is shown in Lemma \ref{lem:finite-set}) below.

So, by \cite{Borkar}, ${\bf r}(i)$ converges to ${\bf r}^{*}$ almost
surely.
\begin{lemma}
\label{lem:r-bounded-Algorithm-4}With Algorithm 1, ${\bf r}(i)$
is always bounded. Specifically, $r_{k}(i)\in[r_{min}-2,r_{max}+2\bar{\lambda}],\forall k,i$,
where $\bar{\lambda}$, as defined before, is the maximal instantaneous
arrival rate, so that $\lambda_{k}'(i)\le\bar{\lambda},\forall k,i$. \end{lemma}
\begin{proof}
We first prove the upper bound $r_{max}+2\bar{\lambda}$ by induction:
(a) $r_{k}(0)\le r_{max}\le r_{max}+2\bar{\lambda}$; (b) For $i\ge1$,
if $r_{k}(i-1)\in[r_{max}+\bar{\lambda},r_{max}+2\bar{\lambda}]$,
then $\tilde{h}(r_{k}(i-1))\le-\bar{\lambda}$. Since $\lambda_{k}'(i)-s_{k}'(i)\le\bar{\lambda}$,
we have $r_{k}(i)\le r_{k}(i-1)\le r_{max}+2\bar{\lambda}$. If $r_{k}(i-1)\in(r_{min},r_{max}+\bar{\lambda})$,
then $\tilde{h}(r_{k}(i-1))\le0$. Also since $\lambda_{k}'(i)-s_{k}'(i)\le\bar{\lambda}$
and $\alpha(i)\le1,\forall i$, we have $r_{k}(i)\le r_{k}(i-1)+\bar{\lambda}\cdot\alpha(i)\le r_{max}+2\bar{\lambda}$.
If $r_{k}(i-1)\le r_{min}$, then \begin{eqnarray*}
r_{k}(i) & = & r_{k}(i-1)+\alpha(i)[\lambda_{k}'(i)-s_{k}'(i)+\tilde{h}(r_{k}(i-1))]\\
 & \le & r_{k}(i-1)+\alpha(i)\{\bar{\lambda}+[r_{min}-r_{k}(i-1)]\}\\
 & = & [1-\alpha(i)]\cdot r_{k}(i-1)+\alpha(i)\{\bar{\lambda}+r_{min}\}\\
 & \le & [1-\alpha(i)]\cdot r_{min}+\alpha(i)\{\bar{\lambda}+r_{min}\}\\
 & = & r_{min}+\alpha(i)\cdot\bar{\lambda}\\
 & \le & \bar{\lambda}+r_{min}\le r_{max}+2\bar{\lambda}.\end{eqnarray*}
The lower bound $r_{min}-2$ can be proved similarly. \end{proof}
\begin{lemma}
\label{lem:finite-set}In Algorithm 1, $w_{0}(i)$ is in a finite
set.\end{lemma}
\begin{proof}
By Lemma \ref{lem:r-bounded-Algorithm-4}, we know that $r_{k}(i)\le r_{max}+2\bar{\lambda},\forall k,i$,
so $T_{k}^{p}(i)\le T_{0}\exp(r_{max}+2\bar{\lambda}),\forall k,i$.
By (\ref{eq:randomize-Tp}), we have $\tau_{k}^{p}(i)\le T_{0}\exp(r_{max}+2\bar{\lambda})+1,\forall k,i$.
Therefore, in state $w_{0}(i)=\{x,((b_{k},a_{k}),\forall k:x_{k}=1)\}$,
we have $b_{k}\le b_{max}$ for a constant $b_{max}$ and $a_{k}\le b_{k}$
for any $k$ such that $x_{k}=1$. So, $w_{0}(i)$ is in a finite
set.
\end{proof}

\subsubsection{Part (ii): Constant step size}

The intuition is the same as part (i). That is, if the constant step
size is small enough, then the algorithm approximately solves problem
$\max_{{\bf r}}G({\bf r};{\bf \lambda})$. Please refer to \cite{longer_version}
for the full proof.

\subsection{Proof of Theorem \ref{thm:collision}}

Let $N'$ be the number of detailed states $(x,z)$'s, and ${\bf u}\in{\cal R}_{+}^{N'}$
be a probability distribution over the detailed states. For convenience
of notation, we use $i=1,2,\dots,N'$ to index the detailed states.
Then $u_{i}$, the $i$-th element of ${\bf u}$, is the probability
of the $i$-th detailed state. Let ${\cal P}$ be the set of $N'$-dimensional
probability distributions, i.e., ${\cal P}:=\{{\bf u}'\in{\cal R}_{+}^{N'}|{\bf 1}^{T}{\bf u}'=1\}$
. Also define a $K\times N'$ matrix $A$ where the element $A_{k,i}=1$
if link $k$ is transmitting its \emph{payload} in the $i$-th detailed
state, and $A_{k,i}=0$ otherwise. Then, $A\cdot{\bf u}$ is the vector
of achieved throughputs of the $K$ links under the distribution ${\bf u}$.
\begin{lemma}
\label{lem:lip}Given $\bar{{\bf \lambda}}$ at the boundary of $\bar{{\cal C}}$,
there exists a constant $b>0$ such that the following holds: For
any $0<\epsilon<1$, if ${\bf u}\in{\cal P}$ and $A\cdot{\bf u}=(1-\epsilon)\bar{{\bf \lambda}}$,
then one can find $\bar{{\bf u}}\in{\cal P}$ such that $A\cdot\bar{{\bf u}}=\bar{{\bf \lambda}}$
and $||{\bf u}-\bar{{\bf u}}||_{1}\le2b\cdot\epsilon$.\end{lemma}
\begin{proof}
Consider the set ${\cal U}:=\{\tilde{{\bf u}}|\tilde{{\bf u}}\in{\cal P},A\cdot\tilde{{\bf u}}=\tilde{\rho}\cdot\bar{{\bf \lambda}}\text{ for some }0\le\tilde{\rho}\le1\}$.
Clearly ${\cal U}$ is a polytope since ${\cal U}$ is defined by
linear constraints. Denote by ${\cal E}({\cal U})$ the set of extreme
points of ${\cal U}$. Then, if ${\bf u}\in{\cal U}$ and satisfies
$A\cdot{\bf u}=(1-\epsilon)\bar{{\bf \lambda}}$, we have \begin{equation}
{\bf u}=\sum_{{\bf y}\in{\cal E}({\cal U})}a_{{\bf y}}{\bf y}\label{eq:u-combination}\end{equation}
where $a_{{\bf y}}\ge0$ and $\sum_{{\bf y}\in{\cal E}({\cal U})}a_{{\bf y}}=1$.
Also, by definition, for any ${\bf y}\in{\cal E}({\cal U})$, $A\cdot{\bf y}=\rho_{{\bf y}}\bar{{\bf \lambda}}$
for some $\rho_{{\bf y}}\in[0,1]$. Then,\[
A\cdot{\bf u}=\sum_{{\bf y}\in{\cal E}({\cal U})}(a_{{\bf y}}A\cdot{\bf y})=\rho_{{\bf u}}\bar{{\bf \lambda}}\]
where $\rho_{{\bf u}}=\sum_{{\bf y}\in{\cal E}({\cal U})}(a_{{\bf y}}\rho_{{\bf y}})=1-\epsilon$.

For any ${\bf u}'\in{\cal U}$, define\begin{eqnarray}
D({\bf u}') & := & \min_{{\bf z}}||{\bf u}'-{\bf z}||_{1}\nonumber \\
 & s.t. & A\cdot{\bf z}=\bar{{\bf \lambda}}.\label{eq:min_distance}\end{eqnarray}
It can be shown that $D({\bf u}')$ is a convex function of ${\bf u}'$.
(Proof: Consider any two ${\bf u}'_{I}$ and ${\bf u}'_{II}$. When
${\bf u}'={\bf u}'_{I}$ (or ${\bf u}'_{II}$), suppose ${\bf z}_{I}$
(or ${\bf z}_{II}$) is an optimal solution of problem (\ref{eq:min_distance}).
That is, $D({\bf u}'_{I})=||{\bf u}'_{I}-{\bf z}_{I}||_{1}$ and $D({\bf u}'_{II})=||{\bf u}'_{II}-{\bf z}_{II}||_{1}$.
Given any $\beta\in[0,1]$, we have \begin{eqnarray}
 &  & \beta\cdot D({\bf u}'_{I})+(1-\beta)\cdot D({\bf u}'_{II})\nonumber \\
 & = & \beta||{\bf u}'_{I}-{\bf z}_{I}||_{1}+(1-\beta)||{\bf u}'_{II}-{\bf z}_{II}||_{1}\nonumber \\
 & \ge & ||\beta({\bf u}'_{I}-{\bf z}_{I})+(1-\beta)({\bf u}'_{II}-{\bf z}_{II})||_{1}\nonumber \\
 & = & ||[\beta{\bf u}'_{I}+(1-\beta){\bf u}'_{II}]-[\beta{\bf z}_{I}+(1-\beta){\bf z}_{II}]||_{1}.\label{eq:middle}\end{eqnarray}

Now consider $D(\beta{\bf u}'_{I}+(1-\beta){\bf u}'_{II})$. Clearly,
$\tilde{{\bf z}}:=\beta{\bf z}_{I}+(1-\beta){\bf z}_{II}$ satisfies
the constraint $A\cdot\tilde{{\bf z}}=\bar{{\bf \lambda}}$. So, $D(\beta{\bf u}'_{I}+(1-\beta){\bf u}'_{II})\le||[\beta{\bf u}'_{I}+(1-\beta){\bf u}'_{II}]-[\beta{\bf z}_{I}+(1-\beta){\bf z}_{II}]||_{1}$.
Combined with (\ref{eq:middle}), we have $D(\beta{\bf u}'_{I}+(1-\beta){\bf u}'_{II})\le\beta\cdot D({\bf u}'_{I})+(1-\beta)\cdot D({\bf u}'_{II})$,
which implies that $D({\bf u}')$ is a convex function of ${\bf u}'$.)

Using this fact and (\ref{eq:u-combination}), we have\begin{equation}
D({\bf u})\le\sum_{{\bf y}\in{\cal E}({\cal U})}(a_{{\bf y}}D({\bf y}))=\sum_{{\bf y}\in{\cal E}({\cal U}),\rho_{{\bf y}}<1}(a_{{\bf y}}D({\bf y}))\label{eq:distance-bound}\end{equation}
where the second step is because if $\rho_{{\bf y}}=1$, then $D({\bf y})=0$.

Define \begin{equation}
b:=\frac{1}{2}\max_{{\bf y}\in{\cal E}({\cal U}),\rho_{{\bf y}}<1}D({\bf y})/(1-\rho_{y}).\label{eq:b}\end{equation}
Since there are a finite number of elements in ${\cal E}({\cal U})$,
we have $0<b<+\infty$. Then \begin{eqnarray*}
D({\bf u}) & \le & \sum_{{\bf y}\in{\cal E}({\cal U}),\rho_{{\bf y}}<1}(a_{{\bf y}}D({\bf y}))\\
 & \le & 2b\cdot\sum_{{\bf y}\in{\cal E}({\cal U}),\rho_{{\bf y}}<1}[a_{{\bf y}}\cdot(1-\rho_{{\bf y}})]\\
 & = & 2b\cdot\sum_{{\bf y}\in{\cal E}({\cal U})}[a_{{\bf y}}\cdot(1-\rho_{{\bf y}})]\\
 & = & 2b[1-\rho_{{\bf u}}]=2b\epsilon.\end{eqnarray*}

Let $\bar{{\bf u}}$ be a solution of (\ref{eq:min_distance}) with
${\bf u}'={\bf u}$, then $||{\bf u}-\bar{{\bf u}}||_{1}=D({\bf u})\le2b\epsilon$. \end{proof}
\begin{lemma}
\label{lem:entropy-bound}Assume that ${\bf u},\bar{{\bf u}}\in{\cal P}$
and ${\bf u}\ne\bar{{\bf u}}$. If $||{\bf u}-\bar{{\bf u}}||_{1}\le2c$
for a constant $c\in(0,1]$, then $|H({\bf u})-H(\bar{{\bf u}})|\le c\cdot[\log(N'/c)+1]$,
where $H({\bf u}):=\sum_{i=1}^{N'}[-u_{i}\log(u_{i})]$ is the {}``entropy''
of the distribution ${\bf u}$.\end{lemma}
\begin{proof}
Let $({\bf u}-\bar{{\bf u}})^{+}$ and $({\bf u}-\bar{{\bf u}})^{-}$
be the positive part and negative part of ${\bf u}-\bar{{\bf u}}$.
Then clearly $||{\bf u}-\bar{{\bf u}}||_{1}=||({\bf u}-\bar{{\bf u}})^{+}||_{1}+||({\bf u}-\bar{{\bf u}})^{-}||_{1}$.
Also, since ${\bf 1}^{T}{\bf u}={\bf 1}^{T}\bar{{\bf u}}=1$, we have
$||({\bf u}-\bar{{\bf u}})^{+}||_{1}=||({\bf u}-\bar{{\bf u}})^{-}||_{1}$.
Therefore, \[
0<||({\bf u}-\bar{{\bf u}})^{+}||_{1}=||({\bf u}-\bar{{\bf u}})^{-}||_{1}=\frac{1}{2}||{\bf u}-\bar{{\bf u}}||_{1}\le c.\]

Let the index sets ${\cal I}_{1}=\{i|i\in\{1,2,\dots,N'\},u_{i}>\bar{u}_{i}\}$
and ${\cal I}_{2}=\{i|i\in\{1,2,\dots,N'\},u_{i}<\bar{u}_{i}\}$.
Define $N_{1}:=|{\cal I}_{1}|,N_{2}:=|{\cal I}_{2}|$. Clearly $N_{1}+N_{2}\le N'$,
and $N_{1},N_{2}>0$. 

Now we show that $H({\bf u})-H(\bar{{\bf u}})\le c\cdot[\log(N'/c)+1]$.
Define the function $f(x):=x\log(1/x)$. First, for $i\in{\cal I}_{1}$,
denote $\delta_{i}=u_{i}-\bar{u}_{i}>0$, then $||({\bf u}-\bar{{\bf u}})^{+}||_{1}=\sum_{i\in{\cal I}_{1}}\delta_{i}$.
Note that $u_{i},\bar{u}_{i},\delta_{i}\in[0,1]$. Since $f(x)$ is
concave, with $u_{i}>\bar{u}_{i}\ge0$, we have $f(u_{i})-f(\bar{u}_{i})\le f(\delta_{i})-f(0)=f(\delta_{i})$.
So\begin{eqnarray}
 &  & \sum_{i\in{\cal I}_{1}}[f(u_{i})-f(\bar{u}_{i})]\le\sum_{i\in{\cal I}_{1}}f(\delta_{i})\nonumber \\
 & \le & \sum_{i\in{\cal I}_{1}}f(\frac{N_{1}}{||({\bf u}-\bar{{\bf u}})^{+}||_{1}})\nonumber \\
 & = & ||({\bf u}-\bar{{\bf u}})^{+}||_{1}\log(\frac{N_{1}}{||({\bf u}-\bar{{\bf u}})^{+}||_{1}})\nonumber \\
 & \le & ||({\bf u}-\bar{{\bf u}})^{+}||_{1}\log(\frac{N'}{||({\bf u}-\bar{{\bf u}})^{+}||_{1}})\label{eq:ineq}\end{eqnarray}
where the second inequality follows from the fact that $f(x)$ is
concave and $\sum_{i\in{\cal I}_{1}}\delta_{i}=||({\bf u}-\bar{{\bf u}})^{+}||_{1}$.

It is easy to show that $f(x)$ is increasing in the range $x\in[0,\exp(-1)]$.
Also, we have $N'\ge3$ (even in the smallest one-link network). Thus
$||({\bf u}-\bar{{\bf u}})^{+}||_{1}/N'\le c/N'\le1/N'<\exp(-1)$.
Therefore, from (\ref{eq:ineq}),\begin{eqnarray}
\sum_{i\in{\cal I}_{1}}[f(u_{i})-f(\bar{u}_{i})] & \le & N'\cdot f(\frac{||({\bf u}-\bar{{\bf u}})^{+}||_{1}}{N'})\nonumber \\
 & \le & N'\cdot f(\frac{c}{N'})=c\log(\frac{N'}{c}).\label{eq:positive-part}\end{eqnarray}

Since $f(x)$ is concave, we have $f(u_{i})-f(\bar{u}_{i})\le(u_{i}-\bar{u}_{i})f'(\bar{u}_{i})$.
For $i\in{\cal I}_{2},$ since $u_{i}<\bar{u}_{i}$ and $f'(\bar{u}_{i})\ge-1$
for any $\bar{u}_{i}\in(0,1]$, we have \begin{eqnarray}
\sum_{i\in{\cal I}_{2}}[f(u_{i})-f(\bar{u}_{i})] & \le & \sum_{i\in{\cal I}_{2}}(\bar{u}_{i}-u_{i})\le c.\label{eq:negative-part}\end{eqnarray}

Using (\ref{eq:positive-part}) and (\ref{eq:negative-part}), we
conclude that $H({\bf u})-H(\bar{{\bf u}})\le c\cdot[\log(N'/c)+1]$.
The same argument shows that $H(\bar{{\bf u}})-H({\bf u})\le c\cdot[\log(N'/c)+1]$.
Therefore the lemma follows.
\end{proof}
Now we are ready to prove Theorem \ref{thm:collision}.

As explained before, ${\bf r}^{*}({\bf \lambda})$ is the vector of
optimal dual variables in the optimization problem (\ref{eq:ME-dual}),
simply written as\begin{eqnarray}
V({\bf \lambda}):= & \max_{{\bf v}\in{\cal P}} & H({\bf v})+\sum_{i}g_{i}v_{i}\nonumber \\
 & s.t. & A\cdot{\bf v}={\bf \lambda}\label{eq:max-entropy}\end{eqnarray}
where $g_{i},i=1,2,\dots,N'$ denote the constants $\log(g(x,z))$
for $(x,z)\in{\cal S}$. Let ${\bf u}^{*}$ be the optimal solution
when ${\bf \lambda}=(1-\epsilon)\bar{{\bf \lambda}}$. Then $V((1-\epsilon)\bar{{\bf \lambda}})=H({\bf u}^{*})+\sum_{i}g_{i}u_{i}^{*}$. 

Consider the case when $\epsilon\le1/b$. By Lemma \ref{lem:lip},
there exists $\bar{{\bf u}}\in{\cal P}$ such that $A\cdot\bar{{\bf u}}=\bar{{\bf \lambda}}$
and $||{\bf u}^{*}-\bar{{\bf u}}||_{1}\le2b\cdot\epsilon\le2$. By
Lemma \ref{lem:entropy-bound}, \begin{equation}
|H({\bf u}^{*})-H(\bar{{\bf u}})|\le b\cdot\epsilon[\log(\frac{N'}{b\cdot\epsilon})+1].\label{eq:ineq3}\end{equation}

Since $\bar{{\bf u}}\in{\cal P}$ satisfies $A\cdot\bar{{\bf u}}=\bar{{\bf \lambda}}$,
we have $V(\bar{{\bf \lambda}})\ge H(\bar{{\bf u}})+\sum_{i}g_{i}\bar{u}_{i}$.
Also, the {}``value function'' $V({\bf \lambda})$ is concave in
${\bf \lambda}$ (\cite{convex-book} page 250), and satisfies $\nabla V({\bf \lambda})=-{\bf r}^{*}({\bf \lambda})$.
Therefore\begin{eqnarray*}
 &  & H(\bar{{\bf u}})+\sum_{i}g_{i}\bar{u}_{i}\le V(\bar{{\bf \lambda}})\\
 & \le & V((1-\epsilon)\bar{{\bf \lambda}})+(\bar{{\bf \lambda}}-(1-\epsilon)\bar{{\bf \lambda}})^{T}[-{\bf r}^{*}((1-\epsilon)\bar{{\bf \lambda}})]\\
 & = & H({\bf u}^{*})+\sum_{i}g_{i}u_{i}^{*}-\epsilon\cdot\bar{{\bf \lambda}}^{T}{\bf r}^{*}((1-\epsilon)\bar{{\bf \lambda}}).\end{eqnarray*}
So, \begin{equation}
\epsilon\cdot\bar{{\bf \lambda}}^{T}{\bf r}^{*}((1-\epsilon)\bar{{\bf \lambda}})\le H({\bf u}^{*})-H(\bar{{\bf u}})+\sum_{i}g_{i}(u_{i}^{*}-\bar{u}_{i}).\label{eq:ineq1}\end{equation}
Denote $G:=\max_{i}|g_{i}|$, then \begin{eqnarray}
|\sum_{i}g_{i}(u_{i}^{*}-\bar{u}_{i})| & \le & \sum_{i}(|g_{i}|\cdot|u_{i}^{*}-\bar{u}_{i}|)\nonumber \\
 & \le & G\sum_{i}|u_{i}^{*}-\bar{u}_{i}|\le2Gb\epsilon.\label{eq:ineq2}\end{eqnarray}

Combining (\ref{eq:ineq1}), (\ref{eq:ineq2}) and (\ref{eq:ineq3}),
we have $\epsilon\cdot\bar{{\bf \lambda}}^{T}{\bf r}^{*}((1-\epsilon)\bar{{\bf \lambda}})\le b\cdot\epsilon[\log(\frac{N'}{b\cdot\epsilon})+1]+2Gb\epsilon$,
which proves (\ref{eq:r-upper-bound-2}).

For the second case when $\epsilon>1/b$, choose an arbitrary $\bar{{\bf u}}\in{\cal P}$
such that $A\cdot\bar{{\bf u}}=\bar{{\bf \lambda}}$. Clearly, $||{\bf u}^{*}-\bar{{\bf u}}||_{1}\le2$,
so the RHS of (\ref{eq:ineq2}) can be replaced by $2G$. Also, $|H({\bf u}^{*})-H(\bar{{\bf u}})|\le\log(N')$
since $H(\cdot)\in[0,\log(N')]$. Using these in (\ref{eq:ineq1})
yields (\ref{eq:r-upper-bound-3}).

\subsection{Simulations of 1-D and 2-D lattice topologies}

Consider a line network (i.e., 1-D lattice network) with 16 links,
where each link conflicts with the nearest 2 links on each side (so,
link $k$ conflicts with 4 other links if it is not at the two ends
of the network), as in Fig. \ref{fig:1d}. Let $\gamma=5,\tau'=10$
and $p_{k}=1/16,\forall k$. In each simulation, we let all links
use the same, fixed payload length $T^{p}=30,50,100,150$, and we
compute the {}``short-term throughput'' of link 8 (in the middle
of the network) every 50 milliseconds (or $50/0.009\approx5556$ slots).
That is, we compute the average throughput of link 8 in each time
window of 50 milliseconds. 

Note that a successful transmission has a length of $\tau'+T^{p}$.
Also note that if the parameter $r_{max}$ in Algorithm 1 satisfies
that $T_{0}\exp(r_{max})=T^{p}$, then it is possible that all links
transmit payload $T^{p}$.

The results are plotted in Fig. \ref{fig:st-line}. We see that the
oscillation in the short-term throughput increases as $T^{p}$ increases,
indicating worsening short-term fairness.%
\begin{figure}
\begin{centering}
\includegraphics[width=8cm]{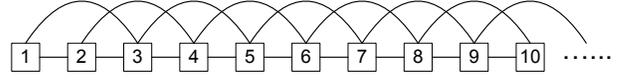}
\par\end{centering}

\caption{\label{fig:1d}Conflict graph of a line network}

\end{figure}
\begin{figure}
\begin{centering}
\subfloat[$T^{p}=30$]{\includegraphics[width=7cm]{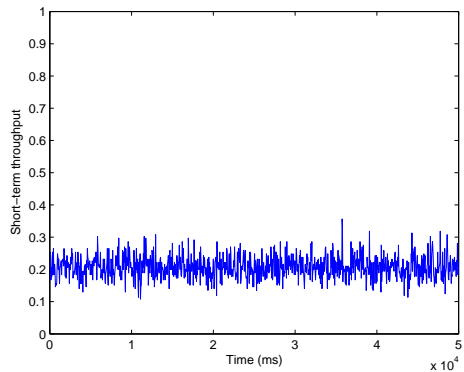}

}
\par\end{centering}

\begin{centering}
\subfloat[$T^{p}=50$]{\includegraphics[width=7cm]{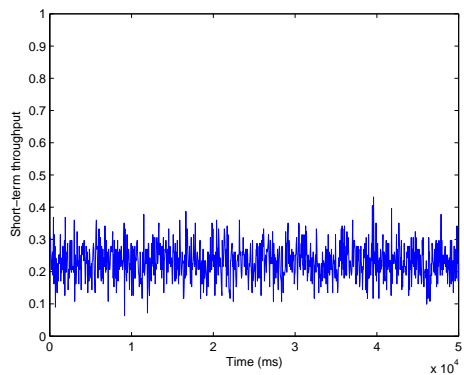}

}
\par\end{centering}

\begin{centering}
\subfloat[$T^{p}=100$]{\includegraphics[width=7cm]{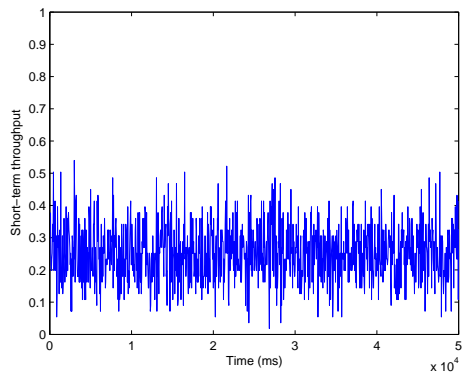}

}
\par\end{centering}

\begin{centering}
\subfloat[$T^{p}=150$]{\includegraphics[width=7cm]{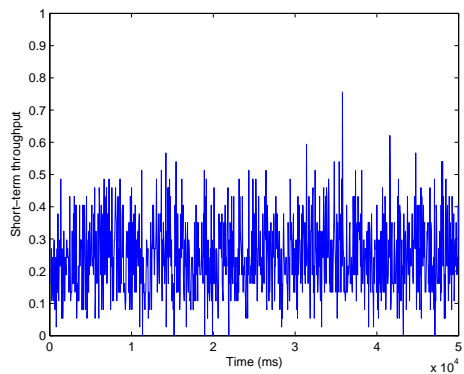}

}
\par\end{centering}

\caption{\label{fig:st-line}Short-term throughput of link 8 in the line network}

\end{figure}

Next we simulate a 2-D lattice network with 5{*}5=25 links in Fig.
\ref{fig:2d}. Each link conflicts with the nearest 4 links around
it (if it's not at the boundary). Similar to the previous simulation,
we use different values of $T^{p}$ and obtain the following short-term
throughput of link 13 (which is at the center of the network) in Fig.
\ref{fig:st-2d}. Again we observe wrosening short-term fairness as
$T^{p}$ increases. Also, the oscillation of link 13's short-term
throughput is greater than link 8 in the line network, since link
13 has 4 conflicting neighbors which do not conflict with each other.
\begin{figure}
\begin{centering}
\includegraphics[width=5cm]{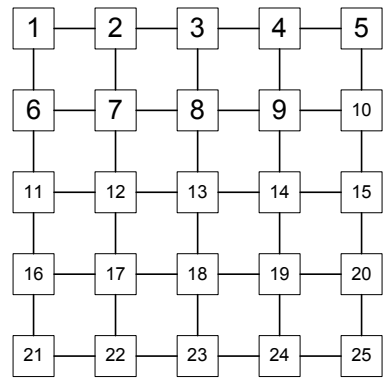}
\par\end{centering}

\caption{\label{fig:2d}Conflict graph of a 2-D lattice network}

\end{figure}
\begin{figure}
\begin{centering}
\subfloat[$T^{p}=25$]{\includegraphics[width=7cm]{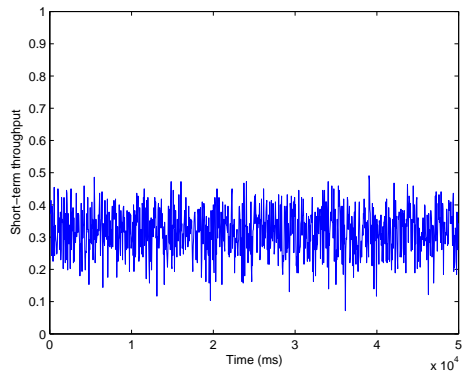}

}
\par\end{centering}

\begin{centering}
\subfloat[$T^{p}=40$]{\includegraphics[width=7cm]{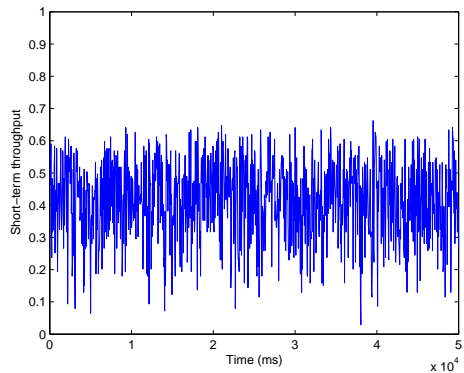}

}
\par\end{centering}

\begin{centering}
\subfloat[$T^{p}=70$]{\includegraphics[width=7cm]{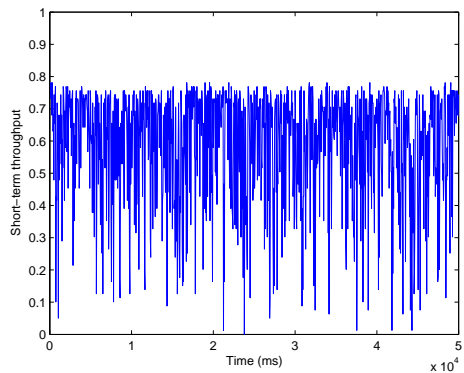}

}
\par\end{centering}

\begin{centering}
\subfloat[$T^{p}=100$]{\includegraphics[width=7cm]{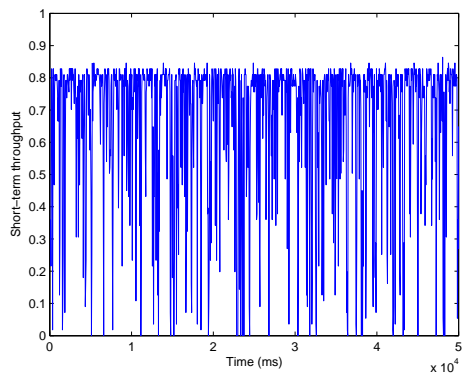}

}
\par\end{centering}

\caption{\label{fig:st-2d}Short-term throughput of link 13 in the 2-D lattice
network}

\end{figure}

However, as mentioned before, it remains an open problem to exact
characterize the relationship between short-term fairness (in particular
the standard deviation of the access delay) and $T^{p}$ (or equivalently,
$r_{max}$) in general topologies. We are interested to further explore
useful methods to quantify the relationship.

\begin{biography}
{Libin Jiang} received his Ph.D. degree in Electrical Engineering
\& Computer Sciences from the University of California at Berkeley
in 2009. Earlier, he received the B.Eng. degree from the University
of Science and Technology of China in 2003, and the M.Phil. degree
from the Chinese University of Hong Kong in 2005. His research interests
include wireless networks, communications and game theory. He received
the David J. Sakrison Memorial Prize in 2010 for outstanding doctoral
research, and the best presentation award in the ACM Mobihoc'09 S3
Workshop.
\end{biography}

\begin{biography}
{Jean Walrand}  received his Ph.D. in EECS from UC Berkeley, where
he has been a professor since 1982. He is the author of An Introduction
to Queueing Networks (Prentice Hall, 1988) and of Communication Networks:
A First Course (2nd ed. McGraw-Hill,1998) and co-author of High Performance
Communication Networks (2nd ed, Morgan Kaufman, 2000). Prof. Walrand
is a Fellow of the Belgian American Education Foundation and of the
IEEE and a recipient of the Lanchester Prize and of the Stephen O.
Rice Prize.
\end{biography}

\end{document}